\documentclass[11 pt]{article}
\date{}

\usepackage[section]{placeins}
\usepackage{graphicx,subfigure}
\usepackage{amsmath,amssymb}
\usepackage{cite}
\usepackage{url}
\usepackage{color}

\newtheorem{theorem}{Theorem}
\newtheorem{lemma}{Lemma}

\newtheorem{remark}{Remark}

\def\dotgt{{\overset \cdot \geq}}
\def\dotlt{{\overset \cdot \leq}}

\newcommand{\med}{\;|\;}
\newcommand{\dff}{\stackrel{\scriptscriptstyle\triangle}{=}}

\newcommand{\bx}{\boldsymbol{x}}
\newcommand{\by}{\boldsymbol{y}}
\newcommand{\bz}{\boldsymbol{z}}
\newcommand{\bc}{\boldsymbol{c}}

\newcommand{\bbe}{\mathbb{E}}

\newcommand{\snr}{{\sf SNR}}
\newcommand{\pr}{{\rm Pr}}

\newcommand{\pe}{P_{\rho}}

\allowdisplaybreaks
\begin{document}
\title{Beacon-Assisted Spectrum Access with Cooperative Cognitive Transmitter and Receiver }

\author{Ali~Tajer~\and~Xiaodong~Wang
\thanks{The authors are with the Department of Electrical Engineering, Columbia University, New York, NY 10027 (email:\{tajer,  wangx\}@ee.columbia.edu).}}
\maketitle

\begin{abstract}
Spectrum access is an important function of cognitive radios for detecting and utilizing spectrum holes without interfering with the legacy systems. In this paper we propose novel cooperative
communication models and show how deploying such cooperations
between a pair of secondary transmitter and receiver assists them in
identifying spectrum opportunities more reliably. These cooperations are facilitated by {\em dynamically} and {\em opportunistically} assigning one of the secondary users as a relay to assist the other one which results in more efficient spectrum hole detection. Also, we investigate the impact of erroneous detection of spectrum holes and thereof missing communication opportunities on the capacity of the secondary channel. The capacity of the secondary users with \emph{interference-avoiding} spectrum access is affected by 1) how effectively the availability of vacant spectrum is sensed by the secondary transmitter-receiver pair, and 2) how correlated are the perceptions of the secondary transmitter-receiver pair about network spectral activity. We show that both factors are improved by using the proposed cooperative protocols. One of the proposed protocols requires explicit information exchange in the network. Such information exchange in practice is prone to wireless channel errors (i.e., is imperfect) and costs bandwidth loss. We analyze the effects of such imperfect information exchange on the capacity as well as the effect of bandwidth cost on the achievable throughput. The protocols are also extended to multiuser secondary networks.
\end{abstract}


\section{Introduction}
\label{introduction}

Cognitive radios have been introduced as a potential solution for alleviating the scarcity of frequency spectrum in overly crowded environments~\cite{Mitola:thesis, Haykin:JSAC05, Tarokh:IT06,
Viswanath:arXiv}. According to the Federal Communication Commission
(FCC), a vast portion of the frequency bands is used only
sporadically and furthermore, the usage varies geographically and
temporally. Such inefficient utilization of spectrum, as well as the
increasing demand for frequency bands by the existing and emerging
wireless applications, motivates opportunistic access to licensed
bands by unlicensed (secondary) users. The notion of cognitive
radios makes it possible to accommodate self-configuring ad-hoc
links within currently established wireless communication
infrastructures. For accessing vacant spectrum bands assigned to
licensed users, the secondary users continuously and actively monitor
the spectrum in order to efficiently make use of the spectrum hole
while avoiding interference with licensed users
(interference-avoiding~\cite{Devroye:Springer07} or interweave
paradigm~\cite{Goldsmith:08_Appear}) or having controlled level of
interference temperature
(interference-controlled~\cite{Devroye:Springer07} or underlay
paradigm~\cite{Goldsmith:08_Appear}).

Related studies on spectrum access can be broadly categorized as
those proposing efficient methods for locating the holes in the
spectrum~\cite{Sahai:online07, Slepian:IT58, Price:IT61} and those
discussing how to optimally allocate the available unused frequency
bands to secondary users~\cite{Tse:JSAC07, Zhao:JSAC07,
Honig:crowncom07, Lifeng:ACM}. The works in the former category aim
at improving the delectability of the unused portions of spectrum by
using physical layer techniques, while those in the latter one
develop media access control (MAC) protocols seeking to maximize
secondary users' network throughput. The channel access protocols we
propose in this paper fall within the first category mentioned
above. These protocols allow the secondary transmitter and receiver to
cooperatively listen to the primary user instead of having independent observations.

The gains yielded for the detection of unused spectrum are due to the
diversity gains introduced by the cooperation. The increase in the
channel capacity is also partly due to the same diversity gain but
is mostly shaped by reducing the uncertainty at the secondary transmitter and receiver nodes about each other's observation of the
spectral activity, which is a by-product of the cooperation between
them. As discussed in~\cite{Jafar:JSAC07}, discrepancy in spectral
activity awareness at the secondary transmitter and receiver, which is due to their spatial separation and/or imperfect channel sensing, incurs a loss in channel capacity.

The underlying idea of the cooperation models developed is to have
a one-time broadcast of a \emph{beacon} message and a one-time
relaying of this message by one of the secondary users such that
\emph{both} enjoy a second order diversity gain in detecting the
beacon message. This beacon message is a codeword within the
codebook of the primary user, reserved specifically for the purpose
of informing the secondary users of the vacant channel.

Besides seeking spectrum opportunities and utilizing them, another
major function of secondary users is to detect the return of
primary users and agilely vacate the channel. The idea of deploying
cooperative diversity schemes and their merits in detecting the
return of the primary users has also been investigated as an
independent problem in~\cite{Li:WCOM07_1, Li:WCOM07_2, Li:JSAC08}.


\section{System Descriptions}
\label{sec:model}

We consider a primary transmitter $T_p$ and a pair of secondary transmitter $T_t$ and receiver $T_r$ users. The secondary users
continuously monitor the channel used by the primary user, seeking
an opportunity to take it over when it is unused. We assume that all
channels between the primary user and the secondary users and
those between the secondary users are quasi-static wireless flat fading channels which remain unchanged during the transmission of a block of symbols and change to independent states afterwards.

We denote the channels between the primary user and the secondary transmitter and receiver by $\gamma_{p,t}$ and $\gamma_{p,r}$
respectively. The channel between the secondary nodes are denoted by
$\gamma_{t,r}$ and $\gamma_{r,t}$. The physical channel between
nodes $i\in\{p,t,r\}$ and $j\in\{t,r\}$ has an instantaneous
realization
\begin{equation}
    \label{eq:model1}
    \gamma_{i,j}=\sqrt{\lambda_{i,j}}\cdot h_{i,j},
\end{equation}
where fading coefficients $h_{i,j}$ are assumed to be independent,
circularly symmetric complex Gaussian random variables
$\mathcal{CN}(0,1)$. The term $\lambda_{i,j}$ accounts for path loss
and shadowing and is given by $\lambda_{i,j}=S_{i,j}d_{i,j}^{-\zeta}$, where $S_{i,j}$ represents the log-normal shadowing effect and
$d_{i,j}$ is the physical distance between nodes $i,j$ and $\zeta>0$
is the path loss exponent.

In~\cite{Sahai:online07} it is shown that compared with energy
detection methods, the detection of vacant spectrum can be
significantly improved through transmitting pilot signals by the
primary user and deploying coherent detection at the secondary users.
Motivated by this fact, in this paper we assume that whenever a channel is released by the primary user, it is announced to the secondary users by having the primary user broadcast a beacon message. It is noteworthy that as eventually we are evaluating the network throughput, it is valid to assume that the secondary users are always willing to access the channel and are continuously monitoring the channel for the beacon message by the primary user.

We consider $N$ consecutive channel uses for the transmission of
the beacon message. The beacon message, sent by the primary user, is
denoted by $\bx_b=[x_b[1],\dots,x_b[N]]^T$ and the received signals
by the secondary transmitter and receiver are denoted by
$\by_t=[y_t[1],\dots,y_t[N]]^T$ and $\by_r=[y_r[1],\dots,y_r[N]]^T$,
respectively. For the direct transmission from the primary user to
the secondary users, as the baseline in our comparisons, we have
\begin{eqnarray}
    \label{eq:model4}
    y_t[n]&=&\gamma_{p,t}x_b[n]+z_t[n],\\
    \mbox{and}\;\;\;\label{eq:model5} y_r[n]&=&\gamma_{p,r}x_b[n]+z_r[n],\; n=1\dots,N,
\end{eqnarray}
where $\bz_t$ and $\bz_r$ denote the zero mean additive white
Gaussian noise terms with variance $N_0$. Also we assume that all
transmitted signals have the same average power, i.e.,
$\bbe[|\bx_b|^2]\leq P_p$, and denote $\rho\dff\frac{P_p}{N_0}$ as the
$\snr$ without fading, pathloss and shadowing. Therefore, the
instantaneous $\snr$ is given by
\begin{equation}
    \label{eq:model3}
    \snr_{i,j}=\rho\lambda_{i,j}\cdot|h_{i,j}|^2.
\end{equation}
Throughout the paper we say that two functions $f(x)$ and
$g(x)$ are \emph{exponentially} equal, denoted by
$f(x)\doteq g(x)$ if
\begin{equation*}
    \lim_{x\rightarrow\infty}\frac{\log f(x)}{\log g(x)}=1.
\end{equation*}
The ordering operators $\dotlt$ and $\dotgt$ are defined
accordingly.

\section{Cooperation Protocols}
\label{sec:protocol}
We assume that the secondary users are informed of the vacancy of
the channels by having the primary user broadcast a beacon
message \emph{a priori} known to the secondary users. The beacon
message is a reserved codeword in the codebook of the primary user
and is dedicated to announcing the availability of the channel for
being accessed by the secondary users. We intend to devise a
cooperation model such that \emph{both} secondary transmitter and
receiver decode the beacon message with a second order diversity
gain.

Cooperative diversity has been studied and developed extensively as
a means for making communication over wireless links more reliable
\cite{Sendonaris:COM03_1, Sendonaris:COM03_2, Laneman:IT04,
Todd:WCOM06}. In cooperative communication a \emph{point-to-point}
link is assisted by an intermediate node (relay) aiming at providing
the intended receiver by additional diversity gains while the
overall resources (frequency bandwidth and power) remain unchanged
compared to the non-cooperative schemes.

In this paper we first develop a novel cooperation model appropriate
for \emph{multicast} transmissions and then show how to adopt it for building cooperative spectrum access schemes. Unlike the
conventional three node cooperative models where the objective is to
achieve a second-order diversity gain at \emph{only} the intended
receiver, we consider broadcasting the same message to two receivers
such that \emph{both} enjoy second-order diversity gains and yet use
the same amount of resources. This cooperation scheme is a combination of opportunistic relay assignment and bit forwarding. While this protocol, like most other existing cooperation models, guarantees performance enhancement in only high enough $\snr$ regimes, we further modify it protocol such that it exerts cooperation only if it is beneficial. This modified protocol, compared to the non-cooperative transmission, will ensure performance improvement over all $\snr$ regimes and all channel realizations.

\subsection{Cooperative Spectrum Access (CSA)}
\label{sec:CSA}

Assume that the beacon message consists of $K$ information bits and $N-K$ parity bits giving rise to the coding rate $\frac{K}{N}$ and requires $N$ channel uses in a non-cooperative transmission. We utilize the regenerative scheme of \cite{Todd:WCOM06} and divide the beacon message into to smaller segments of lengths $K\leq N_1<N$ and $N_2\dff N-N_1$ and define the \emph{level of cooperation} as $\alpha\dff\frac{N_1}{N}$. The transmission of the beacon message is accomplished in two steps:
\begin{enumerate}
  \item The primary user uses only the first $N_1$ channel uses (as opposed to non-cooperative that uses all the $N$ channel uses) to broadcast a reserved codebook in its codebook ($\bx_b'$) when it is willing to release the channel. This reserved codeword is a weaker codeword compared to the original beacon. Meanwhile, the secondary transmitter and receiver try to decode the first segment.
  \item During the remaining $(N-N_1)$ channel uses, any secondary user who has successfully decoded the first segment, immediately constructs $N_2$ additional extra bits and forwards them to the other secondary user. If both happen to decode successfully, both will try to transmit in the second phase and their transmissions collide. Such collision is insignificant and incurs no loss as it only happens when both secondary users have already successfully decoded the beacon. If neither of the secondary users is successful in decoding the first segment of the beacon message, there will be no transmission in the second phase and hence $(1-\alpha)$ portion of the time slot is wasted. Despite such possible waste, as analyses reveal, by cooperation a better overall performance can be achieved. As shown in Section~\ref{sec:diversity}, the combination of opportunistic relay assignment and bit forwarding uses essentially the same amount of resources (power and bandwidth) as in the non-cooperative approach, and yet provides a second-order diversity gain for {\em both} secondary users, whereas if one of the nodes is selected \emph{a priori} and independently of the channel conditions, this node will achieve only a first-order diversity gain.
\end{enumerate}
It is also noteworthy that the CSA protocol can be extended beyond the scope of cognitive networks and can be adopted in any application in which a message is intended to be multicast to several nodes.

\subsection{Opportunistic Cooperative Spectrum Access (OCSA)}
\label{sec:OCSA}

Attaining a higher diversity order, which is a high $\snr$ measure,
regime does not necessarily guarantee performance enhancement in
the low or even moderate $\snr$ regimes. For instance, when either the
source-relay or the relay-destination link is very weak, cooperation
may not be beneficial. When exchanging certain information between the primary and the secondary users is possible, we can overcome this problem by modifying the CSA protocol such that cooperation is deployed only when deemed beneficial. Therefore, the protocol becomes opportunistic in the sense it dynamically decides whether to allow cooperation or not. The cost of this improvement is that the primary and the secondary users need to acquire the quality of their outgoing channels. Acquiring such channel gains is not necessarily feasible in all networks, but those capable of it can benefit from deploying the OCSA protocol instead of the CSA protocol. The discussions on how the required information exchange should take place are provided in Section \ref{sec:feedback}. The cooperation is carried out in two steps:
\begin{enumerate}
  \item The first step is similar to that of the CSA protocol.
  \item All the three nodes (the primary user and the
      pair of secondary transmitter-receiver) step in a competition
      for broadcasting additional party bits for $\bx_b'$ during
      the next $N_2$ channel uses. The competition is carried out
      as follows. We first assign the metrics
      $t_p\dff|\gamma_{p,t}|^2+|\gamma_{p,r}|^2$,
      $t_t\dff|\gamma_{p,t}|^2+|\gamma_{t,r}|^2$ and
      $t_r\dff|\gamma_{p,r}|^2+|\gamma_{r,t}|^2$ to the primary
      user, the secondary transmitter and the secondary receiver,
      respectively. Then the $N_2$ channel uses is allocated to
      \begin{enumerate}
        \item the secondary transmitter if it successfully
            decodes $\bx_b'$ \emph{and} $t_t>\max\{t_p,t_r\}$;
        \item the secondary receiver if it successfully decodes
            $\bx_b'$ \emph{and} $t_r>\max\{t_p,t_t\}$;
        \item to the primary transmitter if \emph{either}
            $t_p>\max\{t_t,t_r\}$ \emph{or} both secondary
            users fail to decode $\bx_b'$.
      \end{enumerate}
\end{enumerate}
As shown in the next sections, including the primary transmitter in the
competition has the advantage of guaranteeing that cooperation is deployed only when it is helpful.

A simple technique for identifying the winner in the second phase in a distributed way (without having a central controller) is to equip each primary and secondary user with a backoff timer, with its initial value set inversely proportional to its corresponding metric $t$. Therefore the timer of the user with the largest metric $t$, goes off sooner and will start relaying. The idea of encoding the backoff timer by channel gain has also been used in~\cite{Bletsas:JSAC06, Ali:ISIT07,Chen:SP07}.

In multiuser primary networks, the central authority of the network responsible for of spectrum coordinations should be in charge of transmitting the beacon message. The reason is that individual users are not fully aware of the network spectral activity and if they act independently there exists the possibility that some user mistakenly announces that a spectrum is vacant while some other primary user is utilizing it. Such strategy will also avoid having multiple primary users transmit beacon messages concurrently, in which case these messages collide and are lost.

\section{Diversity Analysis}
\label{sec:diversity}

In this section we characterize the performance of the proposed
protocols in terms of the probability of erroneous detection of the beacon message, which we denote by $\pe(e)$, at the signal-to-noise ratio $\rho$. The achievable diversity gain which measures how rapidly the error probability decays with increasing $\snr$ is given by
\begin{equation*}
    {\rm diversity\;
    order}=-\lim_{\rho\rightarrow\infty}\frac{\log\pe(e)}{\log\rho}.
\end{equation*}

\subsection{Non-Cooperative Scheme}
\label{sec:div:non_coop}

As a baseline for performance comparisons, we first consider the direct
transmission (non-cooperative) by the primary user over the channels
given by~(\ref{eq:model4})~and~(\ref{eq:model5}). We denote the events that the secondary transmitter and receiver do not decode the beacon message successfully by $e_t$ and $e_r$, respectively. For a coded
transmission with coherent detection, the pairwise error probability
(PEP) that the secondary transmitter erroneously detects the beacon codeword $\bc$ in favor of the codeword $\hat{\bc}$ for the channel realization $\gamma_{p,t}$ is given by \cite[(12.13)]{simon:book}
\begin{equation}
    \label{eq:div:nc1}\pe(e_t\med\gamma_{p,t})\dff
    \pr(\bc\rightarrow\hat{\bc}\med\gamma_{p,t})=Q\bigg(\sqrt{2d\rho|\gamma_{p,t}|^2}\bigg),
\end{equation}
where $d$ is the Hamming distance between $\bc$ and $\hat{\bc}$.
Throughout the diversity order analyses, we will frequently use the
following two results.
\begin{remark}\label{remark:1}
For a real value $a>0$ and the random variable $u>0$ with probability density function $f_U(u)$ we have
\begin{align}
  \nonumber \bbe_u[Q(\sqrt{au})]& = \int_0^\infty\int_{\sqrt{au}}^\infty
  \frac{1}{\sqrt{2\pi}}\;e^{-v^2/2}dv\;f_U(u)\;du\\
  \label{eq:remark1} &=\int_0^{\infty}\pr\bigg(u\leq\frac{v^2}{a}\bigg)\frac{1}{\sqrt{2\pi}}e^{-v^2/2}\;dv.
\end{align}
\end{remark}
\begin{lemma}\label{lemma:3}
For integer $M>0$ and real $k_i>0$,
\begin{equation}
    \label{eq:lemma3}\int_0^{\infty}\prod_{i=1}^M\bigg(1-e^{-k_iv^2/\rho}\bigg)\frac{1}{\sqrt{2\pi}}\;e^{-v^2/2}\;dv
    \doteq\rho^{-M}.
\end{equation}
\end{lemma}
\begin{proof}
See Appendix~\ref{app:lemma3}.
\end{proof}
As a result, from (\ref{eq:div:nc1}), (\ref{eq:remark1}), and (\ref{eq:lemma3}) we can show that for non-cooperative transmission we have
\begin{equation*}
    \pe ^{\rm NC}(e_t)=\bbe_{\gamma_{p,t}}\bigg[Q\bigg(\sqrt{2d\rho |\gamma_{p,t}|^2}\bigg)\bigg]=\rho^{-1},
\end{equation*}
and $\pe ^{\rm NC}(e_r)=\rho^{-1}$.

\subsection{CSA Protocol}
\label{sec:div:CSA}

The transmissions are accomplished in $N=N_1+N_2$ consecutive channel uses, where $N_1$ uses is for the direct transmission, and $N_2$ uses for the cooperation phase. We denote $d_1$ and $d_2$ as the hamming distances of the codewords sent in the first and second phase and we have $d_1+d_2=d$. In the case only one node is successful in the first phase, we denote it by $T_{\cal R}\in\{T_t,T_r\}$ and also define $T_{\bar{\cal R}}\dff\{T_t,T_r\}\backslash T_{\cal R}$. Therefore, we get
\begin{align}
    \label{eq:div:csa8}
    \pe^{\rm C}(e_{\bar{\cal R}}\med \gamma_{p,\bar{\cal R}}, \gamma_{{\cal R},\bar{\cal R}})&=
    Q\bigg(\sqrt{2d_1\rho|\gamma_{p,\bar{\cal R}}|^2+2d_2\rho|\gamma_{{\cal R},\bar{\cal R}}|^2}\bigg),
\end{align}
and for the case that there is no relaying we have
\begin{align}
    \label{eq:div:csa9}
    \pe^{\rm C}(e_t\med \gamma_{p,t})&=
    Q\bigg(\sqrt{2d_1\rho|\gamma_{p,t}|^2}\bigg),\\
    \label{eq:div:csa10}
    \mbox{and}\;\;\;\pe^{\rm C}(e_r\med \gamma_{p,r})&=
    Q\bigg(\sqrt{2d_1\rho|\gamma_{p,r}|^2}\bigg).
\end{align}
\begin{theorem}
\label{th:div_csa}
For sufficiently large $\snr$ values, we have
$\pe^{\rm C}(e_t)\doteq \pe^{\rm C}(e_r)\doteq\frac{1}{\rho^2}$,
$\pe^{\rm C}(e_t)<\pe^{\rm NC}(e_t)$, and $\pe^{\rm C}(e_r)<\pe^{\rm
NC}(e_r)$.
\end{theorem}
\begin{proof} We denote the success and failure of the secondary transmitter in the first phase by $T_t:{\cal S}$ and $T_t:{\cal F}$, respectively.  We also define $T_r:{\cal S}$ and $T_r:{\cal F}$ similarly. By using (\ref{eq:div:csa8})-(\ref{eq:div:csa10}) and by expansion over all different combinations of the statuses of the secondary users, for the secondary transmitter we get
\begin{align}
   \nonumber \pe^{\rm C}&(e_t)\\
    &= \underset{=0}{\underbrace{\pe^{\rm C}(e_t\med T_t:{\cal S},\;T_r:{\cal S})}}\pe^{\rm }(T_t:{\cal S},\;T_r:{\cal S})\\
   \nonumber &+\underset{=0}{\underbrace{\pe^{\rm C}(e_t\med T_t:{\cal S},\;T_r:{\cal F})}}\pe^{\rm}(T_t:{\cal S},\;T_r:{\cal F})\\
   \nonumber &+\pe^{\rm C}(e_t\med T_t:{\cal F},\;T_r:{\cal S})\pe^{\rm}(T_t:{\cal F},\;T_r:{\cal S})\\
   \nonumber &+ \underset{=1}{\underbrace{\pe^{\rm C}(e_t\med T_t:{\cal F},\;T_r:{\cal F})}}\pe^{\rm}(T_t:{\cal F},\;T_r:{\cal F})\\
   \nonumber &= Q\bigg(\sqrt{2d_1\rho|\gamma_{p,t}|^2+2d_2\rho|\gamma_{r,t}|^2}\bigg)\\
   \nonumber &\quad\times Q\bigg(\sqrt{2d_1\rho|\gamma_{p,t}|^2}\bigg)\times \left(1-Q\bigg(\sqrt{2d_1\rho|\gamma_{p,r}|^2}\bigg)\right)\\
   \label{eq:div:csa5} &\quad+ Q\bigg(\sqrt{2d_1\rho|\gamma_{p,t}|^2}\bigg)\times Q\bigg(\sqrt{2d_1\rho|\gamma_{p,r}|^2}\bigg)\\
   \nonumber &= \int_0^{\infty} \underset{\leq P\bigg(d_1|\gamma_{p,t}|^2\leq\frac{v^2}{2\rho}\bigg)P\bigg(d_1|\gamma_{p,r}|^2\leq\frac{v^2}{2\rho}\bigg) }{\underbrace{P\bigg(d_1|\gamma_{p,t}|^2+d_2|\gamma_{r,t}|^2\leq\frac{v^2}{2\rho}\bigg)}}
   \frac{e^{-v^2/2}}{\sqrt{2\pi}}\;dv\\
   \nonumber &\quad\times \int_0^{\infty} P\bigg(d_1|\gamma_{p,t}|^2\leq\frac{v^2}{2\rho}\bigg)
   \frac{e^{-v^2/2}}{\sqrt{2\pi}}\;dv\\
   \nonumber &\quad\times \left[1-\int_0^{\infty} P\bigg(d_1|\gamma_{p,r}|^2\leq\frac{v^2}{2\rho}\bigg)
   \frac{e^{-v^2/2}}{\sqrt{2\pi}}\;dv\right]\\
   \nonumber &\quad+ \int_0^{\infty} P\bigg(d_1|\gamma_{p,t}|^2\leq\frac{v^2}{2\rho}\bigg)
   \frac{e^{-v^2/2}}{\sqrt{2\pi}}\;dv\\
   \label{eq:div:csa11} &\quad\times \int_0^{\infty} P\bigg(d_1|\gamma_{p,r}|^2\leq\frac{v^2}{2\rho}\bigg)
   \frac{e^{-v^2/2}}{\sqrt{2\pi}}\;dv\\
   \nonumber &\leq \underset{\doteq\frac{1}{\rho^2}\;\;\mbox{\footnotesize (from Lemma \ref{lemma:3} and Remark \ref{remark:1})}}{\underbrace{\int_0^{\infty}
   \bigg(1-e^{-\frac{v^2}{2d_1\rho\lambda_{p,t}}}\bigg)
   \bigg(1-e^{-\frac{v^2}{2d_1\rho\lambda_{p,r}}}\bigg)\frac{e^{-v^2/2}}{\sqrt{2\pi}}dv}}\\
   \nonumber &\quad\times \underset{\doteq \frac{1}{\rho}\;\;\mbox{\footnotesize (from Lemma \ref{lemma:3} and Remark \ref{remark:1})}}{\underbrace{\int_0^{\infty}
   \bigg(1-e^{-\frac{v^2}{2d_1\rho\lambda_{p,t}}}\bigg)\frac{e^{-v^2/2}}{\sqrt{2\pi}}\;dv}}\\
   \nonumber &\quad + \underset{\doteq\frac{1}{\rho}\;\;\mbox{\footnotesize (from Lemma \ref{lemma:3} and Remark \ref{remark:1})}}{\underbrace{\int_0^{\infty}
   \bigg(1-e^{-\frac{v^2}{2d_1\rho\lambda_{p,t}}}\bigg)\frac{e^{-v^2/2}}{\sqrt{2\pi}}\;dv}}\\
   \nonumber &\quad \times \underset{\doteq\frac{1}{\rho}\;\;\mbox{\footnotesize (from Lemma \ref{lemma:3} and Remark \ref{remark:1})}}{\underbrace{\int_0^{\infty}
   \bigg(1-e^{-\frac{v^2}{2d_1\rho\lambda_{p,r}}}\bigg)\frac{e^{-v^2/2}}{\sqrt{2\pi}}\;dv}}\doteq \frac{1}{\rho^2}+\frac{1}{\rho^3}\doteq\frac{1}{\rho^2},
\end{align}
where the transition from (\ref{eq:div:csa5}) to (\ref{eq:div:csa11}) follows from Remark 1. Now, by comparing to the non-cooperative case, where we have $\pe^{\rm  NC}(e_t)\doteq\frac{1}{\rho}$, we conclude that for sufficiently large $\snr$, $\pe^{\rm C}(e_t)<\pe^{\rm NC}(e_t)$. The same line of argument holds for the secondary receiver.
\end{proof}

The CSA protocol ensures that for sufficiently high $\snr$ values, the
probability of detecting unused spectrum holes is boosted.
However, the performance in the low $\snr$ regimes, depends on the
instantaneous channel conditions and it might happen that
non-cooperative spectrum access outperforms the cooperative scheme.

\subsection{OCSA Protocol}
\label{sec:div:ocsa}
Like in the CSA protocol, the first $N_1$ channel uses are allocated
for direct transmission from the primary user to the secondary users. During the remaining time, depending on the instantaneous
channel conditions, the primary user might transmit the additional parity bits itself or one of the secondary users might take over such transmission. We denote the transmitting user in the second phase by $T_{\cal R}\in\{T_p,T_t,T_r\}$. Therefore, the transmission in the second phase is given by
\begin{align*}
  y_t[n]&=
  \left\{\begin{array}{cc}
           \gamma_{{\cal R},t}x_{\cal R}[n]+z_t[n] & \mbox{if}\quad T_{\cal R}\neq T_t  \\
           0 &  \mbox{if}\quad T_{\cal R}= T_t
         \end{array}\right.,\\
  \mbox{and}\quad y_r[n]&=
  \left\{\begin{array}{cc}
           \gamma_{{\cal R},r}x_{\cal R}[n]+z_r[n] & \mbox{if}\quad T_{\cal R}\neq T_r  \\
           0 &  \mbox{if}\quad T_{\cal R}= T_r
         \end{array}\right..
\end{align*}
As specified by the protocol, if neither of the secondary users
decodes the beacon message successfully, the primary user will transmit the additional parity bits itself. In this case, transmission of the beacon message is similar to (\ref{eq:model4})-(\ref{eq:model5}) for the
entire channel uses.
\begin{theorem}
\label{th:div_ocsa}
For all values of $\snr$, channel realizations and level of
cooperation, we have $\pe^{\rm OC}(e_t)<\pe^{\rm NC}(e_t)$,
$\pe^{\rm OC}(e_r)<\pe^{\rm NC}(e_r)$, and $\pe^{\rm OC}(e_t)\doteq
\pe^{\rm ONC}(e_r)\doteq\frac{1}{\rho^2}$.
\end{theorem}
\begin{proof}
The probability of missing the beacon message by the secondary transmitter is
\begin{equation}
    \label{eq:improvement1} \pe^{\rm OC}(e_t)=\hspace{-.1in}\sum_{i\in\{p,t,r\}}\pe^{\rm OC}(e_t\med T_{\cal R}=T_i)\pe^{\rm OC}(T_{\cal R}=T_i).
\end{equation}
Note that $\pe^{\rm C}(e_t\med T_{\cal R}=T_t)=0$. Also when $T_{\cal R}=T_r$, according to the protocol we should have $t_r>t_p$, which provides that $|\gamma_{r,t}|>|\gamma_{p,t}|$. Therefore, from (\ref{eq:improvement1}) we get
\begin{align}
    \nonumber&\pe^{\rm OC}(e_t)\\
    \nonumber &=\bbe_{\gamma_{p,t}}\hspace{-.04in}\left[Q\left(\sqrt{2d_1\rho |\gamma_{p,t}|^2+2d_2\rho |\gamma_{p,t}|^2}\right)\right] \hspace{-.03in}\pe^{\rm OC}(T_{\cal R}=T_p) \\
    \nonumber &+ \bbe_{\gamma_{p,t},\gamma_{r,t}}\left[Q\left(\sqrt{2d_1\rho |\gamma_{p,t}|^2+2d_2\rho |\gamma_{r,t}|^2}\right)\right]\\
    \nonumber&\qquad\times\pe^{\rm OC}(T_{\cal R}=T_r)\\
    \nonumber& <\bbe_{\gamma_{p,t}}\left[Q\left(\sqrt{2d\rho |\gamma_{p,t}|^2}\right)\right]\\
    \label{eq:improvement2} &\quad\times[\pe^{\rm OC}(T_{\cal R}=T_p)+\pe^{\rm OC}(T_{\cal R}=T_r)]=\pe^{\rm NC}(e_t).
\end{align}
$d_1$ and $d_2$ are the hamming distances of the codewords sent in the
first and second phase and we have $d_1+d_2=d$. By following the same line of argument we can accordingly show that $\pe^{\rm OC}(e_r)<\pe^{\rm NC}(e_r)$. In order to assess the diversity gain note that
\begin{align}
    \nonumber \pe^{\rm OC}(e_t) &=\pe^{\rm OC}(e_t\med T_t:{\cal F})P(T_t:{\cal F})\\
    \label{eq:improvement3} & \qquad +\underset{=0}{\underbrace{\pe^{\rm OC}(e_t\med T_t:{\cal S})}}P(T_t:{\cal S}).
\end{align}
On the other hand
\begin{align}
    \nonumber \pe^{\rm OC}&(e_t\med T_t:{\cal F})\\
    \nonumber &= \pe^{\rm OC}(e_t\med T_t:{\cal F},T_r:{\cal F})P(T_r:{\cal F})\\
    \nonumber & \qquad +\pe^{\rm OC}(e_t\med T_t:{\cal F},T_r:{\cal S})P(T_r:{\cal S})\\
    \nonumber &\leq P(T_r:{\cal F})+\pe^{\rm OC}(e_t\med T_t:{\cal F},T_r:{\cal S})\\
    \nonumber &\leq \bbe_{\gamma_{p,t}}\bigg[Q\bigg(\sqrt{2d_1\rho |\gamma_{p,r}|^2}\bigg)\bigg]\\
     \nonumber & \qquad + \bbe_{\gamma_{p,t}}\bigg[Q\bigg(\sqrt{2d\rho |\gamma_{p,r}|^2}\bigg)\bigg]\\
     \label{eq:improvement4}&\doteq\frac{1}{\rho}+\frac{1}{\rho}\doteq\frac{1}{\rho}.
\end{align}
Recalling that $P(T_t:{\cal F})\doteq\frac{1}{\rho}$ and using (\ref{eq:improvement3})-(\ref{eq:improvement4}) provides that $\pe^{\rm OC}(e_t)\doteq\frac{1}{\rho}$. A similar argument holds for $\pe^{\rm
OC}(e_r)$.
\end{proof}
\label{sec:evaluations}
\begin{figure}[t]
  \centering
  \includegraphics[width=3.7 in]{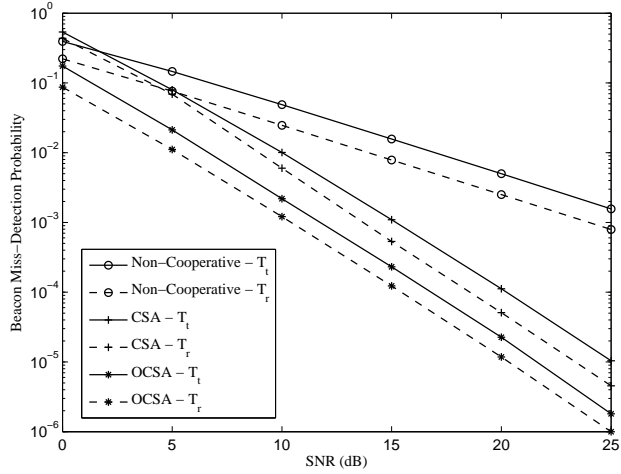}\\
  \caption{Detection error probabilities and diversity gains of non-cooperative, CSA and OCSA schemes.}
  \label{fig:diversity}
\end{figure}
Numerical evaluations of the probabilities of erroneous spectrum
hole detection in different schemes are provided in Fig. \ref{fig:diversity}. These error probabilities for the non-cooperative, the CSA, and the OCSA protocols are provided in (\ref{eq:div:nc1}), (\ref{eq:div:csa11}), and (\ref{eq:improvement2}), respectively. This figure compares the detection performance of the proposed cooperative schemes with that of the non-cooperative scheme, where it demonstrates that the CSA and the OCSA protocols achieve a second-order diversity gain. As expected theoretically, at low $\snr$ regimes the CSA protocol might not outperform the non-cooperative scheme, whereas the OCSA
protocol outperforms the non-cooperative scheme over all $\snr$
regimes. In the evaluations above, we have set $\alpha=0.5$
and have considered channel channel realizations with parameters
$(\lambda_{p,t},\lambda_{p,r},\lambda_{t,r})=(1,2,3)$. In Fig. \ref{fig:diversity}, the error probability of the user $T_i$ for $i\in\{t,r\}$ in the CSA and OCSA protocols is identified by CSA-$T_i$ and OCSA-$T_i$, respectively.

\subsection{False Alarm} \label{sec:false}

Thus far we have examined the improvements attained in detecting the
vacant spectrum holes. In such detection problems, however, another
aspect that should also be taken into consideration is the false alarm probability. In our proposed spectrum access models, false alarm is the event that the secondary users erroneously consider a spectrum band to be idle, i.e., they detect the beacon message while the channel is still used by the primary user. This leads to concurrent transmissions by the primary and secondary users which can potentially harm the transmission
of the primary user. Due to stringent constraints on avoiding concurrent transmissions, it is imperative to study the probability of false alarm as well.

For most detection problems, there exists a tension between the
probability of successful detection and the probability of false
alarm and it is crucial to maintain a good balance between these two
probabilities such that neither of them is sacrificed in favor of
the other one.

In this paper we have translated a detection problem into a
communication problem rested on the basis of exchanging a beacon message. Therefore, assessing the probability of erroneously decoding a non-beacon message in favor of the beacon message has the same nature as that of missing a transmitted beacon and both are quantified in terms of decoding errors. More specifically, the probability that for a channel realization $\gamma_{p,t}$, a non-beacon codeword $\hat \bc$ is mistakenly
decoded as the beacon message denoted by ${\bc}$ (which is the false alarm probability) is given by
\begin{equation}
    \label{eq:false}\pe'(e_t\med\gamma_{p,t})\dff
    \pr(\hat\bc\rightarrow{\bc}\med\gamma_{p,t})=Q\bigg(\sqrt{2d\rho|\gamma_{p,t}|^2}\bigg),
\end{equation}
which is equal to the probability in (\ref{eq:div:nc1}). Therefore,
by following the same lines as in
Sections \ref{sec:div:CSA} and \ref{sec:div:ocsa} it is seen that
the proposed cooperation models will also improve the probability of
false alarm which diminishes as fast as the probability of missing the beacon, enjoying a second order diversity gains.

As a result, deploying the proposed cooperation schemes not only incurs no loss in the probability of false alarm, but also improves it. Hence, by translating the detection problem into a communication problem, the tension between the two aforementioned probabilities is resolved and they can be improved simultaneously.

\section{Capacity Analysis}
\label{sec:capacity}

In this section we analyze the effect of erroneous detection of
unused channel and thereof missing communication opportunities on the
capacity of the secondary users and will show how exploiting the
proposed cooperative diversity protocols enhances the capacity. In the OCSA protocol, each user requires to know the gains of its outgoing channels to the other users. Acquiring such channel gains requires information exchange over the wireless channel. Therefore, the acquired  channel gain estimates are not perfect in practice, which can affect the capacity. In this section, we assume that all channel estimates for the OCSA protocol are perfect and defer analyzing the effect of estimation inaccuracy on the channel capacity to Section \ref{sec:estimation}.

For fading channels there is always a non-zero probability that a
target rate $R$ cannot be supported and there is no meaningful
notion of capacity as a rate of arbitrarily reliable communication.
Therefore, we resort to looking into the notions of outage capacity
for \emph{slow} fading channels and ergodic capacity for \emph{fast}
fading channels.

Capacity of the secondary link is influenced by the spectral activity
of the primary users and the efficiency of the secondary users in
detecting the unused channels. In an earlier study
in~\cite{Jafar:JSAC07}, it is also demonstrated how the secondary user capacity is affected by dissimilar perception of secondary users of primary user's spectral activity in their vicinities,
where it has been shown that more correlated perceptions lead to
higher channel capacity for the secondary link. We will show that the
cooperative protocols are also effective in increasing such
correlation. Lower and upper bounds on the capacity of the secondary
channel, when the received power at the secondary receiver is $P$, are
given by \cite{Jafar:JSAC07}
\begin{align}
  \label{eq:cap4} C^{\rm U}(P)&=\pr(S_t=S_r=1)\log\bigg(1+\frac{P}{\pr(S_t=S_r=1)}\bigg),
\end{align}
and
\begin{align}
  \label{eq:cap5}C^{\rm L}(P)&=\pr(S_t=S_r=1)\log\bigg(1+\frac{P}{\pr(S_t=1)}\bigg)-\frac{1}{T_c}.
\end{align}
where the states $S_t, S_r \in\{0,1\}$, assigned to the secondary
transmitter and receiver respectively, indicate the perception of
the secondary users about the activity of the primary user. $S_t=0$ and
$S_t=1$ mean that the secondary transmitter has sensed the channel
to be busy and idle, respectively, and $S_r$ is defined accordingly.
These state variables retain their states for a period of $T_c$
channel uses and vary to an i.i.d. state afterwards and as seen
above, $T_c$ only affects the lower bound

For further analysis we assume that all channel channels
(i.e., $h_{p,t}, h_{p,r}, h_{t,r}$) follow the same fading
model and therefore the states $S_t$ and $S_r$ have the same time
variations as the secondary channel $h_{t,r}$, which means that all
remain unchanged for $T_c$ channel uses and change to independent
states afterwards. Now, corresponding to different values of $T_c$ we
will have slow and fast fading processes and need to look into
meaningful notions of capacity for each of them.

{\it 1) Fast fading:} Small values of $T_c$ correspond to fast
fading for which a meaningful notion of capacity is given by ergodic
capacity and is obtained by averaging over all channel
fluctuations.
\begin{equation}\label{eq:cap:erg}
    C_{\rm erg}^{\rm U}(\rho)\dff\bbe_{\gamma_{t,r}}\Big[C^{\rm U}(\rho
   \gamma_{t,r})\Big],\; C_{\rm erg}^{\rm L}(\rho)\dff\bbe_{\gamma_{t,r}}\Big[C^{\rm L}(\rho
    \gamma_{t,r})\Big].
\end{equation}

{\it 2) Slow fading:} Corresponding to large values of $T_c\gg 1$ we
consider the $\epsilon$-outage capacity $C_{\epsilon}$ as the
performance measure for which the bounds are given by
\begin{equation}\label{eq:cap:outage}
   \pr\bigg(C^{\rm U(L)}(\rho \gamma_{t,r})<C_{\epsilon}^{\rm U(L)}(\rho)\bigg)\leq
   \epsilon.
\end{equation}
As the ergodic and outage capacities provided
in~(\ref{eq:cap:erg})-(\ref{eq:cap:outage}) depend on the
probability terms $\pr(S_t=1)$ and $\pr(S_t=S_r=1)$, we need to
assess these terms and their variations under different cooperation
models. The states $S_t$ and $S_r$ depend not only the spectral
activity in the vicinity of the secondary transmitter and receiver,
but also on how successfully these users sense this spectral
activity.

We introduce the random variables $\theta_t, \theta_r \in\{0,1\}$ to
account for modeling the spectral activities. $\theta_t=1$
($\theta_r=1$) states that no primary user is using the channel in the vicinity of the secondary transmitter (receiver) and the
channel can be used by the secondary users. $\theta_t=0$ and
$\theta_r=0$ are defined accordingly for busy channels. Based on
the definitions above and those of $S_t$ and $S_r$, which indicate
the perception of secondary users of the activity of the primary
user, we have
\begin{align}
  \label{eq:cap1} &\pr(S_t=1)=\pr(\theta_t=1)(1-\pe(e_t)), \\
  \label{eq:cap6} &\pr(S_t=S_r=1)=\pr(\theta_t=\theta_r=1)\pe(\bar{e}_t,\bar{e}_r).
\end{align}
where $e_t$ and $e_r$ are the erroneous detection events. Note that
the terms $\pr(\theta_t=1)$ and $\pr(\theta_t=\theta_r=1)$ for the
cooperative and non-cooperative schemes are identical and the effect
of cooperation reflects in the terms $\pe(\bar{e}_t)$ and
$\pe(\bar{e}_t,\bar{e}_r)$. In the following two lemmas we show how the CSA and the OCSA protocols improve the probability $\pr(S_t=S_r=1)$.

\begin{lemma}
\label{lemma:1} For sufficiently large $\snr$ we have $\pr^{\rm
C}(S_t=S_r=1)\geq\pr^{\rm NC}(S_t=S_r=1)$.
\end{lemma}
\begin{proof}
We first show that for sufficiently large $\snr$, $\pe^{\rm C}(\bar{e}_r\med\bar{e}_t)\geq\pe^{\rm NC}(\bar{e}_r)$ as follows.
\begin{align}\label{eq:lemma:csa}
    \nonumber \pe^{\rm C}(\bar{e}_r\med\bar{e}_t)&=\underset{=1}{\underbrace{\pe^{\rm C}(\bar{e}_r\med\bar{e}_t,T_r:{\cal S})}}P(T_r:{\cal S})\\
    \nonumber &\qquad +\pe^{\rm C}(\bar{e}_r\med\underset{\equiv T_t:{\cal S},T_r:{\cal F}}{\underbrace{\bar{e}_t,T_r:{\cal F}}})P(T_r:{\cal F})\\
    \nonumber & =  P(T_r:{\cal S})+\pe^{\rm C}(\bar{e}_r\med T_t:{\cal S},T_r:{\cal F})P(T_r:{\cal F})\\
    \nonumber & \geq  \pe^{\rm NC}(\bar{e}_r)P(T_r:{\cal S})\\
    \nonumber &\qquad +\underset{\doteq 1-\frac{1}{\rho^2}}{\underbrace{\pe^{\rm C}(\bar{e}_r\med T_t:{\cal S},T_r:{\cal F})}}P(T_r:{\cal F})\\
    \nonumber&\geq \pe^{\rm NC}(\bar{e}_r)P(T_r:{\cal S})+\underset{\doteq 1-\frac{1}{\rho}}{\underbrace{\pe^{\rm NC}(\bar{e}_r)}}P(T_r:{\cal F})\\
    & =  \pe^{\rm NC}(\bar{e}_r).
\end{align}
On the other hand, from Theorem \ref{th:div_csa} we know that for sufficiently large $\snr$, $\pe^{\rm C}(\bar{e}_t)>\pe^{\rm NC}(\bar{e}_t)$, which in conjunction with (\ref{eq:lemma:csa}) provides that $\pe^{\rm OC}(\bar e_t,\bar e_r)=\pe^{\rm C}(\bar{e}_r\med\bar{e}_t)\pe^{\rm C}(\bar{e}_t)\geq\pe^{\rm NC}(\bar{e}_r)\pe^{\rm NC}(\bar{e}_t)$, which is the desired result.
\end{proof}
\begin{lemma}
\label{lemma:2} For all values of $\snr$ we have
\begin{align*}
\pr^{\rm OC}&(S_t=S_r=1)\geq\\
&\max\bigg\{\pr^{\rm
C}(S_t=S_r=1),\pr^{\rm NC}(S_t=S_r=1)\bigg\}.
\end{align*}
\end{lemma}
\begin{proof} We equivalently show that $\pe^{\rm OC}(\bar e_t,\bar e_r)\geq\max\{\pe^{\rm C}(\bar e_t,\bar e_r),\pe^{\rm NC}(\bar e_t,\bar e_r)\}$. By expanding $\pe^{\rm OC}(\bar e_t,\bar e_r)$ over all different combinations of the statuses of the secondary users (success/failure) we get
\begin{align}
    \nonumber &\pe^{\rm OC}(\bar e_t,\bar e_r) \\
    \nonumber &= \underset{=1}{\underbrace{\pe^{\rm OC}(\bar e_t,\bar e_r\med T_t:{\cal S},\;T_r:{\cal S})}}\pe^{\rm }(T_t:{\cal S},\;T_r:{\cal S})\\
    \nonumber &\quad+\underset{=\pe^{\rm OC}(\bar e_r\med T_t:{\cal S},\;T_r:{\cal F})}{\underbrace{\pe^{\rm OC}(\bar e_t,\bar e_r\med T_t:{\cal S},\;T_r:{\cal F})}}\pe^{\rm}(T_t:{\cal S},\;T_r:{\cal F})\\
    \nonumber &\quad+\underset{=\pe^{\rm OC}(\bar e_t\med T_t:{\cal F},\;T_r:{\cal S})}{\underbrace{\pe^{\rm OC}(\bar e_t,\bar e_r\med T_t:{\cal F},\;T_r:{\cal S})}}\pe^{\rm}(T_t:{\cal F},\;T_r:{\cal S})\\
    \nonumber &\quad+\underset{=\pe^{\rm NC}(\bar e_t)\pe^{\rm NC}(\bar e_r)}{\underbrace{\pe^{\rm OC}(\bar e_t,\bar e_r\med T_t:{\cal F},\;T_r:{\cal F})}}\pe^{\rm}(T_t:{\cal F},\;T_r:{\cal F})\\
    \nonumber &=\pe^{\rm }(T_t:{\cal S},\;T_r:{\cal S})\\
    \nonumber &\quad+\pe^{\rm}(T_t:{\cal S},\;T_r:{\cal F})\\
    \nonumber &\quad\times\left[1-Q\left(\sqrt{2d_1\rho |\gamma_{p,r}|^2+2d_2\rho \max\{\gamma_{p,r}|^2,|\gamma_{t,r}|^2\}}\right)\right]\\
    \nonumber &\quad+\pe^{\rm}(T_t:{\cal F},\;T_r:{\cal S})\\
    \nonumber &\quad\times\left[1-Q\left(\sqrt{2d_1\rho |\gamma_{p,t}|^2+2d_2\rho \max\{\gamma_{p,t}|^2,|\gamma_{t,r}|^2\}}\right)\right]\\
    \nonumber&\quad+\pe^{\rm }(T_t:{\cal F},\;T_r:{\cal F})\\
    \label{eq:lemma:ocsa1}&\quad\times\left[1-Q\left(\sqrt{2d\rho |\gamma_{p,t}|^2}\right)\right]\left[1-Q\left(\sqrt{2d\rho |\gamma_{p,r}|^2}\right)\right].
\end{align}
\begin{figure}[t]
  \centering
  \includegraphics[width=3.7 in]{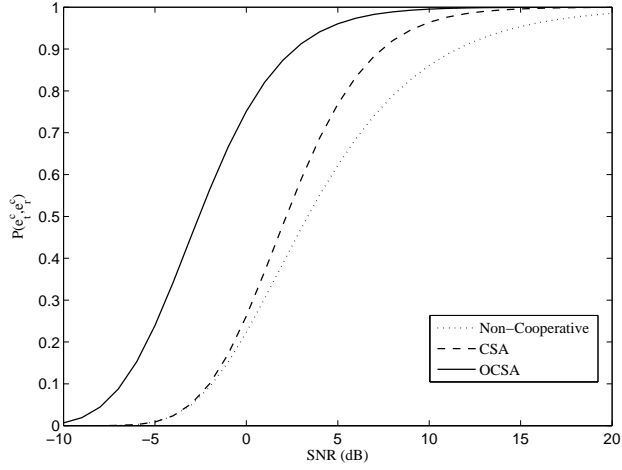}\\
  \caption{Comparing $\pe(\bar{e}_t,\bar{e}_r)$ for different schemes.} \label{fig:correlation}
\end{figure}
By following similar steps we can show that
\begin{align}
    \nonumber \pe^{\rm C}(\bar e_t,\bar e_r) &=\pe^{\rm }(T_t:{\cal S},\;T_r:{\cal S})\\
    \nonumber & +\pe^{\rm}(T_t:{\cal S},\;T_r:{\cal F})\\
    \nonumber & \times \left[1-Q\left(\sqrt{2d_1\rho |\gamma_{p,r}|^2+2d_2\rho |\gamma_{t,r}|^2}\right)\right]\\
    \nonumber
    & +\pe^{\rm}(T_t:{\cal F},\;T_r:{\cal S})\\
    \label{eq:lemma:ocsa2}&\times\left[1-Q\left(\sqrt{2d_1\rho |\gamma_{p,t}|^2+2d_2\rho |\gamma_{t,r}|^2}\right)\right],
\end{align}
Comparing (\ref{eq:lemma:ocsa1}) and (\ref{eq:lemma:ocsa2}) establishes that $\pe^{\rm OC}(\bar e_t,\bar e_r)\geq\pe^{\rm C}(\bar e_t,\bar e_r)$. Also from (\ref{eq:lemma:ocsa1}) we get
\begin{align}
    \nonumber \pe^{\rm OC}&(\bar e_t,\bar e_r) \\
    \nonumber& \geq \pe^{\rm }(T_t:{\cal S},\;T_r:{\cal S})+\pe^{\rm }(T_t:{\cal F},\;T_r:{\cal F})\\
    \nonumber &\quad +\pe^{\rm}(T_t:{\cal S},\;T_r:{\cal F})\left[1-Q\left(\sqrt{2d\rho |\gamma_{p,r}|^2}\right)\right]\\
    \nonumber &\quad+\pe^{\rm}(T_t:{\cal F},\;T_r:{\cal S})\left[1-Q\left(\sqrt{2d\rho |\gamma_{p,t}|^2}\right)\right]\\
    \nonumber &\geq \pe^{\rm }(T_t:{\cal S},\;T_r:{\cal S})+\pe^{\rm }(T_t:{\cal F},\;T_r:{\cal F})\\
    \nonumber &\quad+\pe^{\rm}(T_t:{\cal S},\;T_r:{\cal F})\\
    \nonumber &\quad\times \left[1-Q\left(\sqrt{2d\rho |\gamma_{p,r}|^2}\right)\right]\left[1-Q\left(\sqrt{2d\rho |\gamma_{p,t}|^2}\right)\right]\\
    \nonumber &\quad +\pe^{\rm}(T_t:{\cal F},\;T_r:{\cal S})\\
    \nonumber &\quad \times \left[1-Q\left(\sqrt{2d\rho |\gamma_{p,t}|^2}\right)\right]\left[1-Q\left(\sqrt{2d\rho |\gamma_{p,r}|^2}\right)\right]\\
    \nonumber &\geq \left[1-Q\left(\sqrt{2d\rho |\gamma_{p,t}|^2}\right)\right]\left[1-Q\left(\sqrt{2d\rho |\gamma_{p,r}|^2}\right)\right]\\
    \nonumber &=\pe^{\rm NC}(\bar e_t,\bar e_r)
\end{align}
which completes the proof.
\end{proof}

Figure~\ref{fig:correlation} illustrates numerical evaluations
comparing the term $\pe(\bar{e}_t,\bar{e}_r)$ for different schemes
given in (\ref{eq:lemma:csa})-(\ref{eq:lemma:ocsa2}). As expected, for very low $\snr$ regimes, CSA exhibits no gain over non-cooperative scheme, while OCSA outperforms both CSA and non-cooperative schemes in all $\snr$
regimes. As seen in the figure, the CSA protocol achieves
considerable gain in moderate $\snr$ regimes. For the numerical
evaluations we have assumed the same setup as in Fig. \ref{fig:diversity}.
\begin{theorem}
\label{th:3} For large enough $\snr$ values we have
\begin{eqnarray*}
&C^{\rm{U, C}}_{\rm erg}(\rho)>C^{\rm{U, NC}}_{\rm erg}(\rho),\;
\mbox{\rm and}\;\;\; C^{\rm{L, C}}_{\rm erg}(\rho)>C^{\rm{L,
NC}}_{\rm erg}(\rho),\\
&C^{\rm{U, C}}_{\epsilon}(\rho)>C^{\rm{U,
NC}}_{\epsilon}(\rho),\; \mbox{\rm and}\;\;\; C^{\rm{L,
C}}_{\epsilon}(\rho)>C^{\rm{L, NC}}_{\epsilon}(\rho).\;
\end{eqnarray*}
\end{theorem}
\begin{proof} See Appendix~\ref{app:theorem3}.
\end{proof}
\begin{theorem}
\label{th:4} For all $\snr$ regimes we have
\begin{align*}
  &C^{\rm U, OC}_{\rm
  erg}(\rho)>\max\bigg\{C^{\rm{U, C}}_{\rm erg}(\rho),C^{\rm{U, NC}}_{\rm erg}(\rho)\bigg\},\\
  & C^{\rm L, OC}_{\rm erg}(\rho)>\max\bigg\{C^{\rm{L, C}}_{\rm erg}(\rho),C^{\rm{L, NC}}_{\rm erg}(\rho)\bigg\},\\
  & C^{\rm U, OC}_{\epsilon}(\rho)>\max\bigg\{C^{\rm{U, C}}_{\epsilon}(\rho),C^{\rm{U, NC}}_{\epsilon}(\rho)\bigg\},\\
  \mbox{and}\quad&
  C^{\rm L, OC}_{\epsilon}(\rho)>\max\bigg\{C^{\rm{L, C}}_{\epsilon}(\rho),C^{\rm{L, NC}}_{\epsilon}(\rho)\bigg\}.
\end{align*}
\end{theorem}
\begin{proof} See Appendix~\ref{app:theorem4}.
\end{proof}
\begin{figure}[t]
  \centering
  \includegraphics[width=3.7 in]{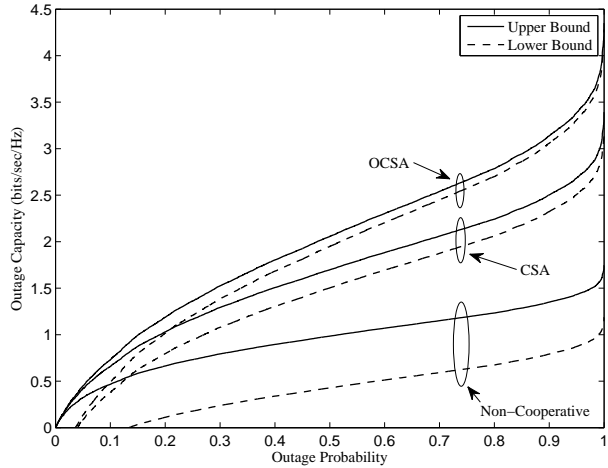}\\
  \caption{Lower and upper bounds on outage capacity for different schemes.}\label{fig:cap_out}
\end{figure}
Figure~\ref{fig:cap_out} demonstrates numerical evaluations of the
outage capacity achieved under different schemes. For these
evaluations we have used the capacity lower and upper bounds given
in (\ref{eq:cap:outage}) and have assumed $\alpha=\frac 1 2$. For
evaluating $\pr(S_t=S_r=1)$ given in (\ref{eq:cap6}) we have assumed
$\pr(\theta_t=\theta_r=1)=0.7$ in both cooperative and
non-cooperative schemes. We also have assumed $\snr$=10 dB, and for the
lower bounds on the outage capacity we have set $T_c=10$. It is observed that outage capacity bound increases as we allow higher outage probability. Also, since the gap between the lower and upper bounds for the cooperative schemes is small, the bounds seem to be tight.

\section{Discussions}
\subsection{Information Exchange}
\label{sec:feedback}

In this section we discuss how the primary and the secondary users acquire the channel state information that they need in the OCSA protocol. The users need to acquire channel gains, which can be facilitated via employing training-based channel estimation schemes and feedback communication. Such estimations and feedback can be carried out reliably via a very low-rate communication (by transmitting few pilot symbols). The  information about channel gains is used only once for initializing the backoff timers, which is during the transition of the channel access from the primary user to the secondary users and is not needed afterwards.

It is noteworthy that acquiring such channel state information is not guaranteed to be feasible in all networks, in which case the secondary users can deploy the CSA protocol that does not require any information exchange.

In the OCSA protocol we need the following channel state information.\\
1) {\it The secondary transmitter should know $|\gamma_{p,t}|$ and the secondary receiver should know $|\gamma_{p,r}|$:}\\
We assume that the primary transmitter is periodically transmitting training symbols to its designated receiver for training purposes. The secondary users can overhear these pilots due to the broadcast nature of the wireless channel and periodically acquire their desired channel states. Obtaining these channel gains is only possible when the primary users transmits training symbols to its designated receiver {\em and} these symbols are known to the secondary users {\em a priori} .\vspace{.1 in}\\
2) {\it The secondary users should know $|\gamma_{t,r}|=|\gamma_{r,t}|$}:\\
This channel gain can be acquired by having the secondary users exchange pilot symbols. This exchange should be performed while the primary user is still active. To avoid harming the communication of primary users, the pilots can be exchanged using ultra-wide band (UWB) communication which allows the secondary users to communicate the low-rate pilots well below the noise level of the primary user. As the secondary users do not have any prior information about the time of the availability of the channel, they should keep exchanging the pilots according to how static the channel $\gamma_{t,r}$ is. Since exchanging pilot sequences requires a very low-rate communication, the process of channel gain estimation can be performed reliably. This level of information exchange can be carried out in all networks. We also assume that the secondary users also exchange the estimates of $|\gamma_{p,t}$ and $|\gamma_{p,r}|$ that they have obtained in Step 1.\vspace{.1 in}\\
3) {\it The primary user should know $|\gamma_{p,t}|$ and $|\gamma_{p,r}|$:} We devise two time slots between the time that primary user finishes the transmission of the beacon, and the time that all users run their backoff timers. During these two time slots, the secondary users feed back the value of $|\gamma_{p,t}|^2+|\gamma_{p,r}|^2$ which they have already obtained during steps 1 and 2. Note that each secondary user will have feedback transmission only if it has successfully decoded the beacon, otherwise its dedicated feedback slot will be wasted. Also, the primary user will not receive any feedback only if both secondary users fail to decode the beacon, in which case its lack of knowledge about the value $|\gamma_{p,t}|^2+|\gamma_{p,r}|^2$ does not harm as the secondary users will not step in the competition phase and any random initializing of the primary users's backoff timer suffices to ensure that the primary user will be transmitting in the second phase.
\subsection{Imperfect Channel Estimation}
\label{sec:estimation}

The OCSA protocol requires the primary to know $|\gamma_{p,t}|^2+|\gamma_{p,r}|^2$, the secondary transmitter to know $|\gamma_{p,t}|^2+|\gamma_{t,r}|^2$, and the secondary receiver to know $|\gamma_{p,r}|^2+|\gamma_{r,t}|^2$. In the capacity analysis in Section \ref{sec:capacity} we have assumed that all users know their corresponding channel state information perfectly.

In practice, however, estimating and feeding back such channel gains is imperfect, which can potentially affect the channel capacity. As discussed in Section \ref{sec:feedback}, these channel gains are used only for the purpose of determining which user should be assigned as relay (to transmit the additional parity bits). Therefore, the effect brought about by imperfect channel estimate is the possibility of selecting a wrong relay. Note that not any imperfect estimation would necessarily lead to selecting a wrong relay, as relay selection is based on the relative order of $\{t_p,t_t,t_r\}$ and not their exact values. For instance, if $t_t=\max\{t_p,t_t,t_r\}$ and we denote the estimates of $\{t_p,t_t,t_r\}$ by $\{\tilde t_p,\tilde t_t,\tilde t_r\}$, there is the chance that $\tilde t_t=\max\{\tilde t_p,\tilde t_t,\tilde t_r\}$ too, in which case the estimation errors do not affect the performance of the protocol and the capacity of the secondary link.

In the OCSA protocol with imperfect channel estimates, for the channel realizations leading to the metrics $\{t_p,t_t,t_r\}$, we denote the probability of selecting a wrong relay by $P_{\rm wr}(t_p,t_t,t_r)$. We also denote the capacity of the secondary link when the right relay is selected by $C^{\rm OC}(P)$ (which is also the capacity with perfect channel estimates) and when a wrong relay is selected by $\widehat{C}^{\rm OC}(P)$, where  $P$ is the power of the received signal by the secondary receiver. Therefore, the capacity when channel estimates are imperfect is given by
\begin{equation}\label{eq:C:imperfect}
    \tilde C^{\rm OC}(P)\dff \bar P_{\rm wr}(t_p,t_t,t_r)\;C^{\rm OC}(P)+P_{\rm wr}(t_p,t_t,t_r)\widehat{C}^{\rm OC}(P).
\end{equation}
Note that $\tilde C^{\rm OC}(P)$ depends on $\{\gamma_{p,t},\gamma_{p,r},\gamma_{t,r}\}$ which is not explicitly expressed in the formulations for the ease of notations. Errors in estimating channel gains ultimately lead to errors in estimating $\{t_p,t_t,t_r\}$ which are used to set the initial values of the backoff timers. We denote such estimation errors by
\begin{equation*}
    w_i\dff \tilde t_i-t_i\quad \mbox{for}\;\;\;i=\{p,t,r\}
\end{equation*}
where $\tilde t_i$ is the estimate of $t_i$, and assume that $w_i\sim\mathcal{N}(0,\sigma^2)$. Therefore, analyzing the effect of imperfect channel estimation on the capacity of the cognitive link translates into analyzing how $\tilde C^{\rm OC}(P)$ and $\sigma^2$ are related, which is provided in the following theorem. We first find an upper bound on the relative capacity loss due to imperfect estimates and then show that one order of magnitude decrease in the estimation errors, translates into one order of magnitude decrease in the capacity loss, e.g., the capacity loss due to estimation noise with variance $\frac{\sigma^2}{100}$ is $\frac{1}{10}$ of the capacity loss due to the estimation noise variance with $\sigma^2$.
\begin{theorem}
\label{th:imperfect}
For i.i.d. channel realizations $\gamma_{p,t}, \gamma_{p,r}, \gamma_{t,r}$ and imperfect channel estimates with estimation noise variance $\sigma^2$,
\begin{enumerate}
  \item the relative capacity loss is bounded as
      \begin{align*}
        &\bbe_{\gamma_{t,r},\gamma_{p,t},\gamma_{p,r}}\left[\frac{C^{\rm OC}(\rho\gamma_{t,r})-\tilde C^{\rm OC}(\rho\gamma_{t,r})}{C^{\rm OC}(\rho\gamma_{t,r})}\right]\\
        &<\frac{1}{3}\;\bbe_{t_p,t_t,t_r}\left[Q\left(\frac{|t_p-t_t|}{\sqrt{2\sigma^2}}\right)+ Q\left(\frac{|t_p-t_r|}{\sqrt{2\sigma^2}}\right)\right],
      \end{align*}
  \item the capacity loss for decreasing values of $\sigma^2$ decays faster than $\sigma$, i.e.,
      \begin{equation*}
        \lim_{\sigma\rightarrow 0} \frac{\log\bbe_{\gamma_{t,r},\gamma_{p,t},\gamma_{p,r}}\left[\frac{C^{\rm OC}(\rho\gamma_{t,r})-\tilde C^{\rm OC}(\rho\gamma_{t,r})}{C^{\rm OC}(\rho\gamma_{t,r})}\right]}{\log\sigma}<\frac{1}{3}
      \end{equation*}
\end{enumerate}
\end{theorem}
\begin{proof}
\begin{enumerate}
  \item We provide the analysis for the case that $\gamma_{p,t}, \gamma_{p,r}, \gamma_{t,r}$ are i.i.d. By following the same line of argument we can find an upper bound for the non-identical distributions as well. When the channel estimates are perfect, we denote the relay selected by the OCSA by $T_\mathcal{R}$. The probability of selecting a wrong relay by some simple manipulations can be expanded as follows.
      \begin{align}\label{eq:wrong_relay1}
        \nonumber &\bbe_{t_p,t_t,t_r}\bigg[ P_{\rm wr}(t_p,t_t,t_r)\bigg]= \\
        &  \hspace{-.1in} \sum_{i\in\{p,t,r\}}\hspace{-.05in}\frac{\bbe_{t_p,t_t,t_r}\bigg[\pe\Big(T_\mathcal{R}\neq T_i\med t_i=\max\{t_p,t_t,t_r\}\Big)\bigg]}{3}.
      \end{align}
      On the other hand we have
      \begin{align}
        \nonumber &\bbe_{t_p,t_t,t_r}\bigg[\pe\Big(T_{\cal R}=T_t\med t_p=\max\{t_p,t_t,t_r\} \Big)\bigg] \\
        \nonumber& = \bbe_{t_p,t_t,t_r}\bigg[\pe(T_t:{\cal S})\underset{1-\pe(T_r:{\cal S})}{\underbrace{\pe(T_r:{\cal F})}}\\
        \nonumber &\hspace{1 in}\times\pe\Big(\tilde t_t>\tilde t_p\med t_p > t_t,t_r\Big)\bigg]\\
        \nonumber& + \bbe_{t_p,t_t,t_r}\bigg[\pe(T_t:{\cal S})\pe(T_r:{\cal S})\\
        \nonumber &\quad\times \pe\Big(\tilde t_t>\tilde t_p\med t_p > t_t,t_r\Big)
        \underset{\leq 1}{\underbrace{\pe\Big(\tilde t_t>\tilde t_r\med t_p > t_t,t_r\Big)}}\bigg]\\
        \nonumber& \leq \bbe_{t_p,t_t,t_r}\bigg[\pe(T_t:{\cal S})\pe\Big(\tilde t_t>\tilde t_p\med t_p > t_t,t_r\Big)\bigg]\\
        \nonumber & \leq \bbe_{t_p,t_t,t_r}\bigg[\pe\Big(\tilde t_t>\tilde t_p\med t_p > t_t,t_r\Big)\bigg]\\
        \label{eq:wrong_relay2}&=\bbe_{t_p,t_t}\bigg[\pe\Big(\tilde t_t>\tilde t_p\med t_p > t_t\Big)\bigg],
      \end{align}
      and similarly
      \begin{align}
        \nonumber \bbe_{t_p,t_t,t_r}\bigg[\pe&\Big(T_{\cal R}=T_r\med t_p > t_t,t_r \Big)\bigg] \leq \\
        \label{eq:wrong_relay3}& \bbe_{t_p,t_r}\bigg[\pe\Big(\tilde t_r>\tilde t_p\med t_p > t_r\Big)\bigg].
      \end{align}
      Also it can be readily verified that
      \begin{align}
        \nonumber &\bbe_{t_p,t_t,t_r}\bigg[\pe\Big(T_{\cal R}=T_p\med t_t > t_p,t_r \Big)\\
        \nonumber & \hspace{.7 in} +\bbe_{t_p,t_t,t_r}\bigg[\pe\Big(T_{\cal R}=T_r\med t_t > t_p,t_r \Big)\bigg]\\
        \label{eq:wrong_relay4} &< \bbe_{t_p,t_t}\bigg[\pe\Big(\tilde t_p>\tilde t_t \med t_t > t_p\Big)\bigg],
      \end{align}
      and similarly
      \begin{align}
        \nonumber &\bbe_{t_p,t_t,t_r}\bigg[\pe\Big(T_{\cal R}=T_p\med t_r > t_p,t_t \Big)\\
        \nonumber & \hspace{.7 in} +\bbe_{t_p,t_t,t_r}\bigg[\pe\Big(T_{\cal R}=T_r\med t_r > t_p,t_t \Big)\bigg]\\
        \label{eq:wrong_relay5} &< \bbe_{t_p,t_t}\bigg[\pe\Big(\tilde t_p>\tilde t_r \med t_r > t_p\Big)\bigg],
      \end{align}
      By noting that $\tilde t_i=t_i+w_i$ where $w_i\sim\mathcal{N}(0,\sigma^2)$, and denoting the probability density functions of $t_i$ by $f(t_i)$, (\ref{eq:wrong_relay1})-(\ref{eq:wrong_relay5}) provide that
      \begin{align}\label{eq:wrong_relay6}
        \nonumber &\bbe_{t_p,t_t,t_r}\bigg[ P_{\rm wr}(t_p,t_t,t_r)\bigg]\\
        \nonumber & <\frac{1}{3}\;\bbe_{t_p,t_t}\bigg[P\Big(\tilde e_t-\tilde e_p> t_p-t_t\med t_p > t_t\Big)\bigg]\\
        \nonumber &\quad +\frac{1}{3}\;\bbe_{t_p,t_t}\bigg[P\Big(\tilde e_p-\tilde e_t> t_t-t_p\med t_t > t_p\Big)\bigg]\\
        \nonumber &\quad +\frac{1}{3}\;\bbe_{t_p,t_r}\bigg[P\Big(\tilde e_r-\tilde e_p> t_p-t_r\med t_p > t_r\Big)\bigg]\\
        \nonumber &\quad +\frac{1}{3}\;\bbe_{t_p,t_r}\bigg[P\Big(\tilde e_p-\tilde e_r> t_r-t_p\med t_r > t_p\Big)\bigg]\\
        \nonumber &= \frac{1}{3}\int_0^\infty\int_0^\infty Q\left(\frac{|t_p-t_t|}{\sqrt{2\sigma^2}}\right)f(t_p)\;dt_p\;f(t_t)\;dt_t\\
        \nonumber &\quad +\frac{1}{3}\int_0^\infty\int_0^\infty Q\left(\frac{|t_p-t_r|}{\sqrt{2\sigma^2}}\right)f(t_p)\;dt_p\;f(t_r)\;dt_r\\
        &=\frac{1}{3}\;\bbe_{t_p,t_t,t_r}\left[Q\left(\frac{|t_p-t_t|}{\sqrt{2\sigma^2}}\right)+ Q\left(\frac{|t_p-t_r|}{\sqrt{2\sigma^2}}\right)\right].
      \end{align}
      Now, from (\ref{eq:C:imperfect}) and (\ref{eq:wrong_relay6}) we get
      \begin{align}\label{eq:imperfect_loss}
        \nonumber &\bbe_{\gamma_{t,r},\gamma_{p,t},\gamma_{p,r}}\left[\frac{C^{\rm OC}(\rho\gamma_{t,r})-\tilde C^{\rm OC}(\rho\gamma_{t,r})}{C^{\rm OC}(\rho\gamma_{t,r})}\right]=\\
        \nonumber &<\bbe_{t_p,t_t,t_r}\bigg[ P_{\rm wr}(t_p,t_t,t_r)\bigg]\\
        &<\frac{1}{3}\;\bbe_{t_p,t_t,t_r}\left[Q\left(\frac{|t_p-t_t|}{\sqrt{2\sigma^2}}\right)+ Q\left(\frac{|t_p-t_r|}{\sqrt{2\sigma^2}}\right)\right],
      \end{align}
      which is the desired result.
  \item We start by showing that
  \begin{equation*}
    \lim_{\sigma\rightarrow 0}\frac{\log \bbe_{t_p,t_t}\left[Q\left(\frac{|t_p-t_t|}{\sqrt{2\sigma^2}}\right)\right]}{\log \sigma}=1.
  \end{equation*}
  By some simplifications we get
  \begin{align}\label{eq:wrong_relay7}
    \nonumber \bbe_{t_p,t_t}&\left[Q\left(\frac{|t_p-t_t|}{\sqrt{2\sigma^2}}\right)\right] =\\
    & \bbe_{\gamma_{p,r},\gamma_{t,r}}\left[Q\left(\frac{\Big||\gamma_{p,r}|^2-|\gamma_{t,r}|^2\Big|}{\sqrt{2\sigma^2}}\right)\right],
  \end{align}
  where we have assumed that $\{\gamma_{p,t},\gamma_{p,r},\gamma_{t,r}\}$ are i.i.d. distributed as $\mathcal{N}(0,\lambda)$ where we have defined $\lambda\dff\lambda_{p,t}=\lambda_{p,r}=\lambda_{t,r}$. Therefore, $|\gamma_{p,r}|^2$ and $|\gamma_{t,r}|^2$ are distributed exponentially with mean $2\lambda$. It can be readily shown that the random variable $X\dff\Big||\gamma_{p,r}|^2-|\gamma_{t,r}|^2\Big|$ is also distributed exponentially with mean $2\lambda$. Therefore, from (\ref{eq:wrong_relay7}) we get
  \begin{equation}\label{eq:wrong_relay8}
    \bbe_{t_p,t_t}\left[Q\left(\frac{|t_p-t_t|}{\sqrt{2\sigma^2}}\right)\right] =\bbe_{X}\left[Q\left(\frac{X}{\sqrt{2\sigma^2}}\right)\right].
  \end{equation}
  On the other hand, we know that $\forall x>0$
  \begin{equation*}
    \frac{1}{\sqrt{2\pi}x}\left(1-\frac{1}{x^2}\right)e^{-x^2/2}\leq Q(x)\leq \frac{1}{\sqrt{2\pi}x}e^{-x^2/2},
  \end{equation*}
  which provides
  \begin{align*}
    &\underset{=\;1}{\underbrace{\lim_{x\rightarrow\infty}\frac{\log\left(\frac{1}{\sqrt{2\pi}x}\left(1-\frac{1}{x^2}\right)e^{-x^2/2}\right)}{\log e^{-x^2/2}}}}\leq\\
    &\qquad \lim_{x\rightarrow\infty}\frac{\log Q(x)}{\log e^{-x^2/2}}\leq \underset{=\;1}{\underbrace{\lim_{x\rightarrow\infty}\frac{\log\left(\frac{1}{\sqrt{2\pi}x}e^{-x^2/2}\right)}{\log e^{-x^2/2}}}},
  \end{align*}
  or equivalently, for asymptotically large values of $x$,
  \begin{equation*}
    Q(x)\doteq e^{-x^2/2}.
  \end{equation*}
  For any given value of $X$, by setting $x=\frac{X}{\sqrt{2\sigma^2}}$, for asymptotically large values of $x$ (or small values of $\sigma$)
  \begin{equation*}
    Q\left(\frac{X}{\sqrt{2\sigma^2}}\right)\doteq e^{-\frac{X^2}{4\sigma^2}}.
  \end{equation*}
  Hence, for asymptotically small values of $\sigma$
  \begin{align*}
    \bbe_X\left[Q\left(\frac{X}{\sqrt{2\sigma^2}}\right)\right]&\doteq \bbe_X\left[e^{-\frac{X^2}{4\sigma^2}}\right]\\
    &=\frac{1}{2\lambda}\int_0^\infty e^{-(x^2/4\sigma^2+x/2\lambda)}dx\\
    &=\frac{1}{\sigma}\cdot\underset{\doteq 1}{\underbrace{\frac{\sqrt{\pi}e^{\sigma^2/4\lambda^2}}{2\lambda} \;Q\left(\frac{\sigma}{\lambda\sqrt{2}}\right)}}\doteq \frac{1}{\sigma},
  \end{align*}
  which in conjunction with (\ref{eq:wrong_relay8}) gives rise to
  \begin{equation}\label{eq:wrong_relay9}
    \bbe_{t_p,t_t}\left[Q\left(\frac{|t_p-t_t|}{\sqrt{2\sigma^2}}\right)\right]= \bbe_X\left[Q\left(\frac{X}{\sqrt{2\sigma^2}}\right)\right] \doteq \frac{1}{\sigma}.
  \end{equation}
  Similarly, we can show that
  \begin{equation}\label{eq:wrong_relay10}
    \bbe_{t_p,t_t}\left[Q\left(\frac{|t_p-t_r|}{\sqrt{2\sigma^2}}\right)\right] \doteq \frac{1}{\sigma}.
  \end{equation}
  Hence,
  \begin{equation}\label{eq:wrong_relay11}
    \bbe_{t_p,t_t,t_r}\left[Q\left(\frac{|t_p-t_t|}{\sqrt{2\sigma^2}}\right)+ Q\left(\frac{|t_p-t_r|}{\sqrt{2\sigma^2}}\right)\right]\doteq\frac{1}{\sigma}.
  \end{equation}
  (\ref{eq:imperfect_loss}) and (\ref{eq:wrong_relay11}) provide that
  \begin{align*}
    \lim_{\sigma\rightarrow 0}&\frac{\log\bbe_{\gamma_{t,r},\gamma_{p,t},\gamma_{p,r}}\left[\frac{C^{\rm OC}(\rho\gamma_{t,r})-\tilde C^{\rm OC}(\rho\gamma_{t,r})}{C^{\rm OC}(\rho\gamma_{t,r})}\right]}{\log \sigma} \\
    & <  \lim_{\sigma\rightarrow 0}\frac{\log\bbe_{t_p,t_t,t_r}\left[Q\left(\frac{|t_p-t_t|}{\sqrt{2\sigma^2}}\right)+ Q\left(\frac{|t_p-t_r|}{\sqrt{2\sigma^2}}\right)\right]}{3\log\sigma}\\
    & =  \frac{1}{3},
  \end{align*}
  which is the desired result.
\end{enumerate}
\end{proof}
\begin{figure}
  \centering
  \includegraphics[width=3.7 in]{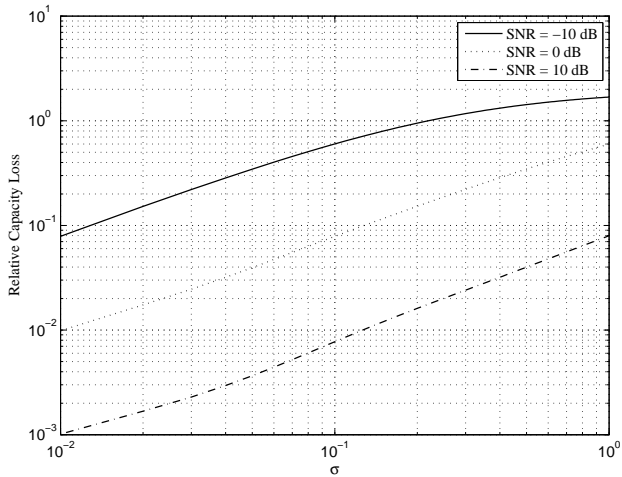}\\
  \caption{Relative capacity loss versus channel estimation noise accuracy.}\label{fig:imperfect}
\end{figure}

The numerical evaluation of $\bbe_{t_p,t_t,t_r}\left[Q\left(\frac{|t_p-t_t|}{\sqrt{2\sigma^2}}\right)+ Q\left(\frac{|t_p-t_r|}{\sqrt{2\sigma^2}}\right)\right]$ versus the variance of noise estimation ($\sigma^2)$ is depicted in Fig. \ref{fig:imperfect} for for three different $\snr$ values. We have considered the setup $\lambda_{p,t}=\lambda_{p,r}=\lambda_{t,r}=1$ where we can see that when the noise variance is $\sigma^2=0.01$ the capacity loss will be less than $\% 1$ at $\rho=0$ dB, which is a negligible loss. Also, it is observed that for any fixed value of $\sigma^2$, increasing the $\snr$ results in less capacity loss. This is justified by noting that for fixed noise level, more powerful signal are less prone to estimation errors.
\subsection{Throughput Analysis}
\label{sec:throughput}

Channel capacity, being an intrinsic characteristic of the wireless channel, is not influenced by how effectively the channel is utilized or by how much the MAC-level coordinations (e.g., backoff timers and information exchange in the OCSA protocol) cost. Therefore, in order to furnish a fair comparison between the OCSA protocol and the non-cooperative scheme and to incorporate the MAC-level costs of this protocol, we assess the achievable throughputs.

We have developed a {\em throughput} analysis, which rests on the basis of our {\em capacity} analysis, to account for two types of throughput losses incurred by the OCSA protocol. One loss is due to the time waste imposed by the backoff timers and the other loss is due to the required feedback from the secondary users to the primary user.

As as running the backoff timers and exchanging information take place only one time, if the secondary users access the channel for a period sufficiently larger than those of the backoff timers and the feedback communications, the throughput will be very close to the capacity and these throughput losses will have a negligible effect.

In order to formulate the throughput, which we denote by $R$, we define $T_{\rm CR}$ as the duration that the secondary link will use the channel and we denote the time dedicated to the feedback communication by $T_{\rm FB}$. Finally, we denote the initial value of the backoff timer of user $T_i$ by $T^i_{\rm BT}$ which is set as
\begin{equation*}
    T_{\rm BT}^i\dff\frac{\beta}{t_i}\quad\mbox{for}\;\;i\in\{p,t,r\},
\end{equation*}
where $\beta$ is a constant and its unit depends on the unit of $t_i$. Since $\{t_i\}$ are scalars, $\beta$ has the units of time.
Therefore, for any channel realization $\{\gamma_{p,t},\gamma_{p,r},\gamma_{t,r}\}$, when the relay is user $T_i$, the time delay due to backoff timers is $T_{\rm BT}^i$ and the throughput is found as
\begin{align}\label{eq:throughput}
    \nonumber R(\gamma_{p,t},\gamma_{p,r},\gamma_{t,r}) & = \frac{T_{\rm CR}}{T_{\rm CR}+T_{\rm FB}+T_{\rm BT}^i}\cdot\tilde C^{\rm OC}(\rho\gamma_{t,r})\\
    & =  \frac{T_{\rm CR}}{T_{\rm CR}+T_{\rm FB}+\frac{\beta}{t_i}}\cdot\tilde C^{\rm OC}(\rho\gamma_{t,r}),
\end{align}
where we have taken into account that there is a one-time feedback transmission and backoff timer activation. (\ref{eq:throughput}) suggests that when the duration that the secondary users access the channel is considerably longer than the time required for feedback, i.e., $T_{\rm CR}\gg T_{\rm FB}$, the throughput loss due to feedback will be negligible. However, the same argument does not apply to the effect of backoff timers as for very weak channels (small $t_i$), $\frac{\beta}{t_i}$ can become a non-negligible factor compared to $T_{\rm CR}$. In the following theorem, we assess the {\em average} throughput loss over the possible channel realizations.
\begin{theorem}\label{th:throughput}
The average throughput loss due to information exchange and backoff timers is upper bounded by
\begin{align*}
    \bbe_{\gamma_{t,r},\gamma_{p,t},\gamma_{p,r}}&\left[\frac{\tilde C^{\rm OC}(\rho\gamma_{t,r})-R(\gamma_{p,t},\gamma_{p,r},\gamma_{t,r})}{\tilde C^{\rm OC}(\rho\gamma_{t,r})}\right]\\
    &\leq \frac{w_1}{1+w_1}+(1+w_1)\left(e^{\frac{w_2}{1+w_1}}-1\right),
\end{align*}
where
\begin{equation*}
    w_1\dff\frac{T_{\rm FB}}{T_{\rm CR}}\quad\mbox{and}\quad w_2\dff\frac{\beta}{T_{\rm CR}\lambda_{p,t}}.
\end{equation*}
\end{theorem}
\begin{proof}
First note that for the node $T_i$ selected as the relay, $t_i\geq t_p=|\gamma_{t,p}|^2+|\gamma_{t,r}|^2>|\gamma_{t,p}|^2$. Therefore, from (\ref{eq:throughput}) we get
\begin{equation*}
    \frac{R(\gamma_{p,t},\gamma_{p,r},\gamma_{t,r})}{\tilde C^{\rm OC}(\rho\gamma_{t,r})}  \geq \frac{1}{1+w_1+w_2\cdot\frac{1}{|h_{p,t}|^2}}.
\end{equation*}
Therefore,
\begin{align*}
    &\bbe_{\gamma_{t,r},\gamma_{p,t},\gamma_{p,r}}\left[\frac{R(\gamma_{p,t},\gamma_{p,r},\gamma_{t,r})}{\tilde C^{\rm OC}(\rho\gamma_{t,r})}\right] \\
    &\geq  \bbe_{\gamma_{p,t}}\left[\frac{T_{\rm CR}}{T_{\rm CR}+T_{\rm FB}+\frac{\beta}{{|\gamma_{p,t}|^2}}}\right]\\
    &=\bbe_{h_{p,t}}\left[\frac{1}{1+w_1+w_2\cdot\frac{1}{|h_{p,t}|^2}}\right]\\
    \\
    & =  \left[\frac{(1+w_1)e^{-t}+w_2e^{\frac{w_2}{1+w_1}}Ei\left(-\frac{w_2+(1+w_1)t}{1+w_1}\right)}{(1+w_1)^2}\right]_{t=0}^{t=\infty}\\
    & =  \frac{1}{1+w_1}-w_2e^{\frac{w_2}{1+w_1}}\int_{\frac{w_2}{1+w_1}}^\infty\frac{e^{-t}}{t}\;dt\\
    & \geq  \frac{1}{1+w_1}-w_2e^{\frac{w_2}{1+w_1}}\left(1-e^{-\frac{w_2}{1+w_1}}\right)\frac{1+w_1}{w_2}\\
    & =  \frac{1}{1+w_1}-(1+w_1)\left(e^{\frac{w_2}{1+w_1}}-1\right),
\end{align*}
which establishes the desires result.
\end{proof}
\begin{figure}
  \centering
  \includegraphics[width=3.7 in]{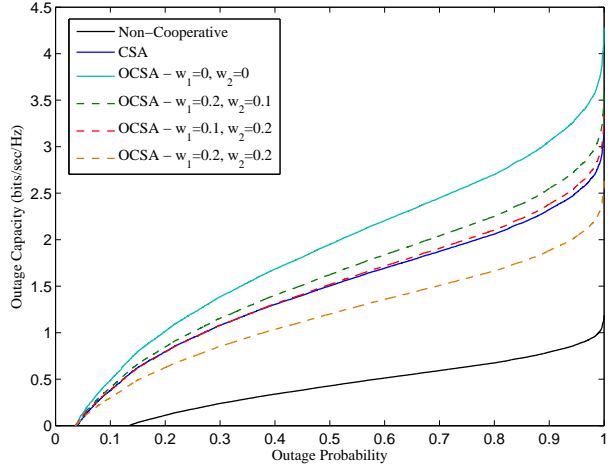}\\
  \caption{Throughput comparison in different schemes.}\label{fig:throughput}
\end{figure}
From the result above, it is consequently concluded that for the appropriate choice of $\beta$, i.e., by setting
\begin{equation*}
    \beta\ll T_{\rm CR}\lambda_{p,t}\quad \Rightarrow w_2\quad \ll 1
\end{equation*}
we can make the effect of the backoff timers very negligible which provides that $w_2\rightarrow 0$ and
\begin{equation*}
    \bbe_{\gamma_{t,r},\gamma_{p,t},\gamma_{p,r}}\left[\frac{\tilde C^{\rm OC}(\rho\gamma_{t,r})-R(\gamma_{p,t},\gamma_{p,r},\gamma_{t,r})}{\tilde C^{\rm OC}(\rho\gamma_{t,r})}\right]\leq \frac{w_1}{1+w_1}.
\end{equation*}
As for the effect of feedback, there will be a loss which is upper bounded by $\frac{T_{\rm CR}}{T_{\rm CR}+T_{\rm FB}}$. This loss will be also marginal if the period that the cognitive link is active is sufficiently larger than the small period of the feedback communication, i.e., $T_{\rm CR}\gg T_{\rm FB}$.

Figure \ref{fig:throughput} depicts the achievable throughput by the OCSA under different assumptions on feedback loads and backoff timer settings and compares them to the throughput achievable via the non-cooperative scheme and the CSA protocol. The solid curves represent the lower bounds on the capacity found in Section \ref{sec:capacity}. For the non-cooperative and the CSA schemes, due to having no MAC-layer bandwidth loss, the throughput is essentially equivalent to the capacity, whereas for the OCSA protocol the throughput becomes equal to the capacity when we set $w_1=w_2=0$. The dashed curves illustrate the throughput of the OCSA protocol for the different choices of $w_1,w_2$. We have considered a similar setup as in the numerical evaluations of Fig. \ref{fig:cap_out}. The results provide that for the choice of $w_1<0.15$ and $w_2<0.15$, the throughput of the OCSA scheme is higher than that of the CSA scheme and as $w_1$ and $w_2$ exceed these levels, the throughput of the OCSA protocol falls below that of the CSA protocol. Note that in practical scenarios, it is reasonable to assume that $w_1,w_2<0.15$ as $w_1=\frac{T_{\rm FB}}{T_{\rm CR}}$, which is the ratio of the length of the one-time feedback packet to the length of the information packets and is close t zero. Also, $w_2$ depends on the constant factor $\beta$ of the timers which can be chosen such that keeps the value of $w_2$ arbitrarily small.

\section{Multiuser Network}
\label{sec:multi}

\subsection{Multiuser CSA Protocol (MU-CSA)}
The cooperation protocols proposed in Section~\ref{sec:protocol} and
the subsequent analyses in Section~\ref{sec:diversity} consider a
cognitive network with a single pair of secondary transmitter and
receiver. In this section we provide a direction for generalizing
the CSA protocol to a cognitive network consisting of multiple secondary transmitter-receiver pairs. The objective is to provide {\em all} secondary users with the diversity gain $2M$ for decoding the beacon message.

We consider a multiuser network consisting of $M$ pairs of secondary transmitters and receivers denoted by $(T_t^1,T_r^1),\dots,(T_t^M,T_r^M)$. The physical channels between
the primary user and the $m$th secondary transmitter and receiver
are represented by $\gamma^m_{p,t}$, $\gamma^m_{p,r}$, respectively,
and $\gamma^m_{t,r}$ denotes the channel between $(T_t^m,T_r^m)$ for
$m=1,\dots, M$. As in Section~\ref{sec:model} for the channel
realizations we have
\begin{equation}
  \label{eq:multi1}
  \gamma^m_{i,j}=\sqrt{\lambda^m_{i,j}}\;h^m_{i,j},
\end{equation}
where $h^m_{i,j}$, the fading coefficients, are independent complex
Gaussian $\mathcal{CN}(0,1)$ random variables and $\lambda^m_{i,j}$
represent pathloss and shadowing effects.In the CSA protocol, a secondary user acts as relay based on its success in decoding the first segment of the beacon. The MU-CSA protocol consists of two steps:
\begin{enumerate}
  \item The primary user broadcasts the first segment of the beacon message during the initial $0<\alpha<1$ portion of the time slot as in the CSA and OCSA protocols and all the secondary transmitters and receivers listen to the message and at the end of the transmission try to decode the message.
  \item During the remaining $(1-\alpha)$ portion of the time slot, \emph{all} secondary users who have successfully decoded the first segment of the beacon construct the additional parity bits and broadcast them. Note that there is the possibility that multiple secondary users have been successful in the first phase. Since all these users broadcast {\em identical} information packets, their concurrent transmissions are not considered as collision and rather can be deemed as a distributed multiple-antenna transmission. For maintaining fairness in terms of the amount of resources consumed, the transmission power of individual secondary users acting as relay should not exceed $\frac{P_p}{2M}$ in order to keep the aggregate transmission power below $P_p$. There is the issue of synchronization between the successful relays, which they can ensure by using the beacon they have received as the synchronization reference.
\end{enumerate}

\subsection{Diversity Analysis}
Next we show that in a multiuser cognitive network with the MU-CSA
protocol, all secondary users in the enjoy the diversity gain of $2M$ in detecting the beacon message. We define the set $\mathcal{T}=\{\tilde T_1, \dots, \tilde T_{2M}\}= \{T_t^1,\dots,T_t^M\}\bigcup\{T_r^1,\dots,T_r^M\}$. Also we define ${\cal M}\dff\{1,\dots,2M\}$ and denote the channel between the secondary users $\tilde T_m$ and $\tilde T_n$ by $\tilde\gamma_{m,n}$ and the channel between $T_p$ and $\tilde T_m$ by $\tilde\gamma_{p,m}$. The probability of missing the beacon message by the secondary user $\tilde T_m$ is
\begin{align*}
    &\pe^{\rm C}(e_{\tilde T_m}) = \hspace{-.1 in}\sum_{\underset{A\neq \emptyset}{A\subseteq{\cal M}\backslash\{m\}}}\hspace{-.1 in}\pe^{\rm C}(e_{\tilde T_m}\med \forall i\in A, \tilde T_i:{\cal S}; \forall j\notin A, \tilde T_j:{\cal F})\\
    &\hspace{1.1 in}\times\pe(\forall i\in A, \tilde T_i:{\cal S}; \forall j\notin A, \tilde T_j:{\cal F})\\
    & \hspace{1.1 in}+ \underset{=1}{\underbrace{\pe^{\rm C}(e_{\tilde T_m}\med \forall i,\;\tilde T_i:{\cal F})}}\pe(\forall i,\;\tilde T_i:{\cal F})\\
    & = \sum_{\underset{A\neq \emptyset}{A\subseteq{\cal M}\backslash\{m\}}} \underset{\doteq \frac{1}{\rho^{|A|+1}}\;\;\mbox{\footnotesize (from Lemma \ref{lemma:3})}}{\underbrace{Q\left(\sqrt{2d_1\rho |\tilde\gamma_{p,m}|^2+2d_2\frac{\rho}{2M}\sum_{i\in A}|\tilde\gamma_{i,m}|^2}\right)}}\\
    &\hspace{1.1in}\times\underset{\doteq\frac{1}{\rho^{2M-|A|}}\;\;\mbox{\footnotesize (from Lemma \ref{lemma:3})}}{\underbrace{\prod_{j\notin A}Q\left(\sqrt{2d_1\rho |\tilde\gamma_{p,j}|^2}\right)}}\\
    &\hspace{1.1 in}\times \prod_{i\in A}\bigg[1-\underset{\frac{1}{\rho}\;\;\mbox{\footnotesize (from Lemma \ref{lemma:3})}}{\underbrace{Q\left(\sqrt{2d_1\rho |\tilde\gamma_{p,i}|^2}\right)}}\bigg]\\
    & \hspace{1.1 in}+ \prod_{i\in\{1,\dots,2M\}}\underset{\doteq\frac{1}{\rho}\;\;\mbox{\footnotesize (from Lemma \ref{lemma:3})}}{\underbrace{Q\left(\sqrt{2d_1\rho |\tilde\gamma_{p,j}|^2}\right)}}\\
    &\doteq \frac{1}{\rho^{2M+1}}+\frac{1}{\rho^{2M}}\doteq \frac{1}{\rho^{2M}}
\end{align*}
Therefore, all of the secondary users enjoy a diversity order $2M$ in decoding the beacon message.

The channel capacity can be obtained similarly as in
Section~\ref{sec:capacity} for each secondary pair when only that pair accesses the channel. For finding the channel capacity of a specific pair of secondary users we need to find 1) the probability of detecting the beacon by the secondary transmitter and 2) the probability that both of secondary transmitter and receiver are successful in decoding to beacon. The lower and upper bounds on the channel capacity are given in (\ref{eq:cap4}) and (\ref{eq:cap5}) and the capacity for the link between $(T_t^m,T_r^m)$ can be found by plugging the values of
\begin{equation*}
    \pr(S_{T_t^m}=1)=\pr(\theta_{T_t^m}=1)\pe^{\rm
    C}(\bar e_{T_t^m}),
\end{equation*}
and
\begin{equation*}
    \pr(S_{T_t^m}=S_{T_r^m}=1)=\pr(\theta_{T_t^m}=\theta_{T_r^m}=1)\pe^{\rm
    C}(\bar e_{T_t^m},\bar e_{T_r^m}).
\end{equation*}
\begin{figure}[t]
  \centering
  \includegraphics[width=3.7 in]{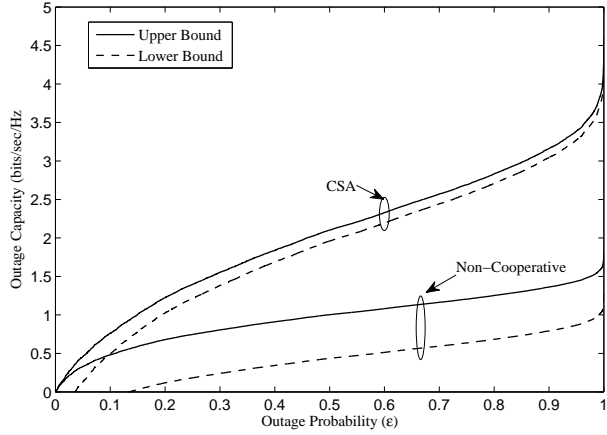}\\
  \caption{Outage capacity in a cognitive network with 5 pairs of users.}
  \label{fig:cap_multi_out}
\end{figure}
Numerical evaluations for the lower and upper bounds on the
secondary user capacity in a cognitive network consisting of 5
pairs of users are shown in Fig.~\ref{fig:cap_multi_out}, which
depicts the outage capacity. For this evaluation we have used the
same setup as that in Section~\ref{sec:evaluations}. By comparing
these results with those of a network with one pair of secondary users it is seen that more gain in terms of channel capacity is achievable for the multiuser network which is due to opportunistic selection of the cognitive pair. As the number of secondary users pairs increases, the quality of the best secondary user is expected to become better which justifies the additional gains observed in the multiuser cognitive networks.

\section{Conclusions}
\label{sec:conclusion}

In this paper we have proposed opportunistic cooperation protocols for supervised spectrum access in cognitive networks where the primary and secondary users are only allowed to have orthogonal access to the channel and cannot coexist concurrently. The primary user broadcasts a beacon message upon releasing the channel and a pair of secondary transmitter and receiver cooperatively detect the beacon message. The cooperation protocols are devised such that by a single-time relaying both secondary transmitter and receiver enjoy second-order diversity gains in detecting the beacon message. We have also quantified the effect of erroneous detection of vacant channels on the secondary channel capacity and show the advantage of the proposed protocols for achieving higher secondary channel capacity. One of the proposed protocols imposes information exchange and bandwidth costs, where we have assessed the effect of imperfect information exchange on the capacity and the loss due to the bandwidth cost on the achievable throughput. The cooperation models are also extended to multi-user cognitive networks as well. 
\appendix

\section{Proof of Lemma~\ref{lemma:3}}
\label{app:lemma3}

First we show that for the function $A(M,n)\dff\sum_{m=0}^M
{M\choose m}m^n(-1)^m$, $A(M,n)=0$ for $\forall M>n\geq 0$, and
$A(M,M)\neq 0$. We provide the proof by induction:\\
{\bf 1)} For $n=0$ and $\forall M>0$ we have $A(M,0)=\sum_{m=0}^M
{M\choose m}(-1)^m=(1+(-1))^M=0$.\\
{\bf 2)} Assumption: $\forall M>n>0$ we have
$A(M,n)=0$.\\
{\bf 3)} Claim: $\forall M>n+1$ we show that $A(M,n+1)=0$.
\begin{align*}
  A(M,n+1) &= \sum_{m=0}^M {M\choose m}m^{n+1}(-1)^m \\
  &=-M\sum_{m=0}^{M-1}{M-1\choose m}(m+1)^n(-1)^{m}\\
  &= -M\sum_{k=0}^n{n\choose k}\bigg[\sum_{m=0}^{M-1}{M-1\choose
  m}m^k(-1)^m\bigg]\\
  &= -M\sum_{k=0}^n{n\choose
  k}A(\underset{M-1>n\geq k}{\underbrace{M-1,k}})=0.
\end{align*}
Mow, we show that $A(M,M)\neq 0$. Again, by induction we have:\\
{\bf 1)} For $M=1$, $A(1,1)=-1$.\\
{\bf 2)} Assumption: $A(M,M)\neq0$.\\
{\bf 3)} Claim: $A(M+1,M+1)\neq0$.
\begin{align*}
  A(M+1,& M+1)\\
  &=\sum_{m=0}^{M+1} {M+1\choose m}m^{M+1}(-1)^m \\
  &=-(M+1)\sum_{m=0}^M{M\choose m}(m+1)^M(-1)^{m}\\
  &= -(M+1)\sum_{k=0}^M{M\choose
  k}\bigg[\sum_{m=0}^{M}{M\choose m}m^k(-1)^m\bigg]\\
  &= -(M+1)\bigg[\underset{\neq0}{\underbrace{A(M,M)}}+\sum_{k=0}^{M-1}{M\choose
  k}\underset{=0}{\underbrace{A(M,k)}}\bigg]\\
  & \neq0.
\end{align*}
Knowing the above result and considering the following Taylor series
expansion
\[\big(1+kx\big)^{-\frac{1}{2}}=\sum_{n=0}^{\infty}\frac{(2n)!}{2^{2n}(n!)^2}(-kx)^n,\]
and noting that
$\frac{1}{\sqrt{2\pi}}\int_0^{\infty}e^{-\frac{av^2}{2}}dv=\frac{1}{2}a^{-\frac{1}{2}}$
we get
\begin{align*}
  \int_0^{\infty}&\bigg(1-e^{-kv^2/\rho}\bigg)^M\frac{1}{\sqrt{2\pi}}e^{-v^2/2}\;dv \\ &=
  \sum_{m=0}^M {M\choose m}(-1)^m\frac{1}{\sqrt{2\pi}}\int_0^{\infty}e^{-\frac{v^2}{2}(1+2km/\rho)}dv \\
  &= \frac{1}{2}\sum_{m=0}^M {M\choose
  m}(-1)^m\bigg(1+\frac{2km}{\rho}\bigg)^{-1/2}\\
  &=\frac{1}{2}\sum_{m=0}^M {M\choose
  m}(-1)^m\sum_{n=0}^{\infty}\frac{(2n)!}{2^{2n}(n!)^2}(-1)^n\bigg(\frac{2km}{\rho}\bigg)^n\\
  &= \frac{1}{2}\sum_{n=0}^{\infty}\rho^{-n}\frac{(2n)!}{2^{n}(n!)^2}(-k)^n\sum_{m=0}^M {M\choose
  m}m^n(-1)^m\\
  &=
  \frac{1}{2}\sum_{n=0}^{\infty}\rho^{-n}\frac{(2n)!}{2^{n}(n!)^2}(-k)^nA(M,n)\\
  &=
  \frac{1}{2}\sum_{n=M}^{\infty}\rho^{-n}\frac{(2n)!}{2^{n}(n!)^2}(-k)^nA(M,n)\doteq\rho^{-M},
\end{align*}
therefore
\begin{align*}
   \rho^{-M}&\doteq\int_0^{\infty}\prod_{i=1}^M\bigg(1-e^{-\min_i k_iv^2/\rho}\bigg)\frac{1}{\sqrt{2\pi}}
   \;e^{-v^2/2}\;dv\;\\
   &\leq \int_0^{\infty}\prod_{i=1}^M\bigg(1-e^{-k_iv^2/\rho}\bigg)\frac{1}{\sqrt{2\pi}}\;e^{-v^2/2}\;dv\;\\
   &\leq \int_0^{\infty}\prod_{i=1}^M\bigg(1-e^{-\min_i
   k_iv^2/\rho}\bigg)\frac{1}{\sqrt{2\pi}}\;e^{-v^2/2}\;dv\\
   &\;\doteq
   \rho^{-M}
\end{align*}
which concludes the lemma.

\section{Proof of Theorem~\ref{th:3}}
\label{app:theorem3}

We show that for any $P>0$, $C^{\rm U, C}(P)>C^{\rm U, NC}(P)$ and
$C^{\rm L, C}(P)>C^{\rm L, NC}(P)$ which consequently proves the
theorem. First, we remark that $p\log(1+\frac{a}{p})$ is non-decreasing in $p$ since
\begin{eqnarray*}
  \frac{\partial }{\partial p}\;C^{\rm U}(a)
  =\frac{\partial}{\partial p}\; p\log\bigg(1+\frac{a}{p}\bigg)\geq
  0.
\end{eqnarray*}
The inequality above holds because for the function $f(u)\dff\log
u-(1-\frac{1}{u})$, $f(1)=0$ and for $u\geq 1$, $\frac{\partial
f(u)}{\partial u}=\frac{1}{u}-\frac{1}{u^2}\geq 0$ and therefore
$f(u)\geq 0$ for $u\geq 1$. Therefore, $C^{\rm U}(P)$ as given in (\ref{eq:cap4}) is non-decreasing
in the probability $\pr(S_t=S_r=1)$. By using Lemma~\ref{lemma:1} it
is concluded that for sufficiently large $\snr$ $C^{\rm U,
C}(P)>C^{\rm U, NC}(P)$, which in turn, establishes the result for
the upper bound. For the lower bound, as shown in (\ref{eq:lemma:csa}), for sufficiently large $\snr$, $\pe^{\rm C}(\bar{e}_r\med\bar{e}_t)\geq\pe^{\rm NC}(\bar{e}_r)$ as shown
\begin{equation*}
    \frac{\pe^{\rm C}(\bar{e}_t,\bar{e}_r)}{\pe^{\rm NC}(\bar{e}_t)\pe^{\rm NC}(\bar{e}_r)}>\frac{\pe^{\rm C}(\bar{e}_t)}{\pe^{\rm NC}(\bar{e}_t)},
\end{equation*}
or equivalently
\begin{equation}
   \label{eq:th3:proof1}\frac{\pr^{\rm C}(S_t=S_r=1)}{\pr^{\rm NC}(S_t=S_r=1)}>\frac{\pr^{\rm C}(S_t=1)}{\pr^{\rm
   NC}(S_t=1)}.
\end{equation}
Therefore, for the lower bound we get that for sufficiently large $\snr$
\begin{align}
  \nonumber
  C^{\rm L,C}(P)&=\pr^{\rm
  C}(S_t=S_r=1)\log\bigg(1+\frac{P}{\pr^{\rm
  C}(S_t=1)}\bigg)-\frac{1}{T_c}\\
  \nonumber&>\pr^{\rm
  C}(S_t=S_r=1)\\
  \label{eq:th3:proof3} &\quad\times\log\bigg(1+\frac{P}{\frac{\pr^{\rm
  C}(S_t=S_r=1)}{\pr^{\rm
  NC}(S_t=S_r=1)}\pr^{\rm
  NC}(S_t=1)}\bigg)-\frac{1}{T_c}\\
  \nonumber &> \pr^{\rm NC}(S_t=S_r=1)\\
  \nonumber&\quad\times\log\bigg(1+\frac{P}{\pr^{\rm
  NC}(S_t=1)}\bigg)-\frac{1}{T_c}\\
  \label{eq:th3:proof4}&= C^{\rm L, NC}(P),
\end{align}
where (\ref{eq:th3:proof3}) follows from (\ref{eq:th3:proof1}) and the transition from where (\ref{eq:th3:proof3}) to (\ref{eq:th3:proof4}) follows from the fact that $p\log(1+\frac{a}{p})$ is increasing in $p$.
\section{Proof of Theorem~\ref{th:4}}
\label{app:theorem4}

As shown in Appendix~\ref{app:theorem3}, $C^{\rm U}(P)$ is non-decreasing in $\pr(S_t=S_r=1)$. By using Lemma~\ref{lemma:2} it is readily
verified that for the function $C^{\rm U}(P)$
\begin{equation*}
  C^{\rm U,OC}(P)>\max\bigg\{C^{\rm{U,C}}(P),C^{\rm{U,NC}}(P)\bigg\}.
\end{equation*}
For the lower bound, we deploy the same approach as in the proof of
Theorem~\ref{th:3} and first show that
\begin{eqnarray}
  \label{eq:th4:proof1}\frac{\pr^{\rm OC}(S_t=S_r=1)}{\pr^{\rm C}(S_t=S_r=1)}&\geq&\frac{\pr^{\rm OC}(S_t=1)}{\pr^{\rm
  C}(S_t=1)},\\
  \mbox{and}\;\;\;\label{eq:th4:proof2}\frac{\pr^{\rm OC}(S_t=S_r=1)}{\pr^{\rm NC}(S_t=S_r=1)}&\geq&
  \frac{\pr^{\rm OC}(S_t=1)}{\pr^{\rm
  NC}(S_t=1)}.
\end{eqnarray}
The above inequalities can be proven by substituting and manipulating the expansions of $\pr^{\rm C}(e_t)$ and $\pr^{\rm OC}(e_t)$ given in Theorems \ref{th:div_csa}, \ref{th:div_ocsa} and the expansions of $\pr^{\rm C}(\bar e_t,\bar e_r)$ and $\pr^{\rm OC}(\bar e_t,\bar e_r)$  (\ref{eq:lemma:ocsa1}) and (\ref{eq:lemma:ocsa2}). After establishing the above inequalities, the last step is similar to that of Appendix \ref{app:theorem3}.

\renewcommand\url{\begingroup\urlstyle{rm}\Url}

{\small \bibliographystyle{IEEEtran}
\bibliography{IEEEabrv,DSA_Coop}}

\end{document}